\documentclass[aps,longbibliography,prd,onecolumn,nofootinbib]{revtex4}

\usepackage[a4paper,top=2cm,bottom=2cm,left=2cm,right=2cm]{geometry}
\usepackage[utf8]{inputenc}
\usepackage{graphicx}
\usepackage{xcolor}
\usepackage{amsmath,amssymb,amsfonts,amsthm}
\usepackage{mathtools,mathrsfs}
\usepackage{enumerate}
\usepackage{dsfont}
\usepackage{scrextend}
\usepackage{float}
\usepackage{comment}
\usepackage{verbatim}
\usepackage[default]{gillius}

\usepackage{nomencl}
\makenomenclature

\newcommand\connsym[2]{\{ \substack{#1\\#2} \}} 

\newcommand{\Lie}{\mathcal{L}_V} 
\usepackage{bm}
\usepackage[page,toc,titletoc,title]{appendix}

\usepackage[
		format=hang, 
		justification=RaggedRight, 
		textfont=it, 
		figurename = Figure,
	]{caption}

\usepackage{listings}
\usepackage{color}

\definecolor{codegreen}{rgb}{0,0.4,0}
\definecolor{codegray}{rgb}{0.3,0.3,0.3}
\definecolor{codepurple}{rgb}{0.38,0,0.42}

\lstdefinestyle{mystyle}{
	commentstyle=\color{codegreen},
	keywordstyle=\color{magenta},
	numberstyle=\tiny\color{codegray},
	stringstyle=\color{codepurple},
	basicstyle=\footnotesize,
	breakatwhitespace=false,
	breaklines=true,
	captionpos=b,
	keepspaces=true,
	numbers=left,
	numbersep=5pt,
	showspaces=false,
	showstringspaces=false,
	showtabs=false,
	tabsize=4
}
\usepackage[colorlinks,citecolor=blue,linkcolor=blue,urlcolor=blue]{hyperref}

\newcounter{ThmCounter}
\newtheorem{thm}{Theorem}[section]

\newtheorem{lemma}[thm]{Lemma}
\newtheorem{remark}{Remark}[section]
\newtheorem{definition}[thm]{Definition}

\usepackage{tikz}
\usetikzlibrary{arrows.meta,matrix}
\usetikzlibrary{decorations.pathmorphing}
\tikzset{
    boson/.style={draw,ultra thick},
    fermion/.style={draw,ultra thick,decorate,decoration=snake},
    vertex/.style={circle,thick,draw=black,fill=black!10,minimum size=18pt,inner sep=3pt},
    empty/.style={circle,thick,draw=black,minimum size=18pt,inner sep=3pt},
    outin/.style={rectangle,thick,minimum size=18pt,inner sep=3pt},
}

\makeatletter 
\newcommand{\oset}[3][0ex]{%
	\mathrel{\mathop{#3}\limits^{
			\vbox to#1{\kern-2\ex@
				\hbox{$\scriptstyle#2$}\vss}}}}
\makeatother

\newcommand{\restr}[2]{{
		\left.\kern-\nulldelimiterspace
		#1 
		\vphantom{\big|}\,
		\right|_{#2}
	}}

\begin{document}

\nomenclature[01]{$M$}{Manifold}
\nomenclature[02]{$g_{ab}$}{Metric field}
\nomenclature[03]{$\Gamma^{a}_{bc}$}{Connection field}
\nomenclature[04]{$\nabla$}{Covariant derivative associated with $\nabla$}
\nomenclature[05]{$\mathcal{D}$}{Levi-Civita covariant derivative associated with metric $g$}
\nomenclature[06]{$Q_{abc}$}{Non-metricity tensor $Q = \nabla g$}
\nomenclature[07]{$L^{a}{}_{bc}$}{Disformation tensor, $L^a{} _{bc}=\frac{1}{2} Q^a{}_{bc}-Q_{(b}{}^a{}_{c)}$}
\nomenclature[08]{$\connsym{a}{bc}$}{Levi-Civita Christoffel symbols}
\nomenclature[09]{$\mathcal{R}(g)$}{Curvature associated with Levi-Civita connection from metric $g$}
\nomenclature[10]{$R(\Gamma)$}{Curvature associated with connection $\Gamma$}
\nomenclature[11]{$\bm{L}$}{Lagrangian $n$-form}
\nomenclature[12]{$\varepsilon$}{Completely antisymmetric symbol (Levi-Civita symbol)}
\nomenclature[13]{$\bm{\epsilon}$}{Volume form on manifold, $\bm{\epsilon} = \sqrt{-g}\varepsilon$}
\nomenclature[14]{$\mathcal{L}$}{Lagrangian density, $\bm{L} = \varepsilon \mathcal{L}$}
\nomenclature[15]{$\bm{Q}_N$}{Noether charge}
\nomenclature[16]{$\bm{J}$}{Noether current}
\nomenclature[17]{$\bm{\omega}$}{Presymplectic form}
\nomenclature[18]{$\bm{\Omega}$}{Symplectic form}
\nomenclature[19]{$P^{a}{}_{bc}$}{Non-metricity conjugate, defined by $\sqrt{-g} P^a{}_{bc} = \frac{\partial \mathcal{L}}{\partial Q_{a}{}^{bc}}$}

\title{Wald's entropy in Coincident General Relativity}

\author{Lavinia Heisenberg}\email{lavinia.heisenberg@phys.ethz.ch}
\author{Simon Kuhn}\email{simkuhn@phys.ethz.ch}
\author{Laurens Walleghem}\email{laurensw@econ.uio.no}
\affiliation{\small
\mbox{Institute for Theoretical Physics, ETH Z\"{u}rich, 8093 Z\"{u}rich, Switzerland}}
\affiliation{\small
\mbox{Institute for Theoretical Physics, Heidelberg University, Philosophenweg 16, 69120 Heidelberg, Germany}}

\begin{abstract}
The equivalence principle and its universality enables the geometrical formulation of gravity. In the standard formulation of General Relativity \'a la Einstein, the gravitational interaction is geometrized in terms of the spacetime curvature. However, if we embrace the geometrical character of gravity, two alternative, though equivalent, formulations of General Relativity emerge in flat spacetimes, in which gravity is fully ascribed either to torsion or to non-metricity. The latter allows a much simpler formulation of General Relativity oblivious to the affine spacetime structure, the Coincident General Relativity. The entropy of a black hole can be computed using the Euclidean path integral approach, which strongly relies on the addition of boundary or regulating terms in the standard formulation of General Relativity. A more fundamental derivation can be performed using Wald's formula, in which the entropy is directly related to Noether charges and is applicable to general theories with diffeomorphism invariance. In this work we extend Wald's Noether charge method for calculating black hole entropy to spacetimes endowed with non-metricity. Using this method, we show that Coincident General Relativity with an improved action principle gives the same entropy as the well-known entropy in standard General Relativity. Furthermore the first law of black hole thermodynamics holds and an explicit expression for the energy appearing in the first law is obtained. 
\end{abstract}

\maketitle

\vspace{1cm}
\section{Introduction and overview}

A general spacetime is endowed with an arbitrary affine connection and a metric. Such spacetimes allow the formulation of gravity as a gauge theory of translations, in spirit with Yang-Mills theories. The field strengths of the connection give the curvature and the torsion of such a spacetime and the metric further allows the introduction of the non-metricity. Embracing fully the geometrical nature, the very same theory of General Relativity (GR) can be attributed either to curvature, torsion or non-metricity, giving rise to a geometrical trinity of gravity \cite{jimenez2019}.

One can derive the thermodynamical properties of a black hole from a grand partition function given by the path integral over the gravitational fields. In the standard formulation of General Relativity based on curvature, the computation of the Euclidean action of a Schwarzschild black hole strongly relies on the introduction of the Gibbons- Hawking-York (GHY) boundary term. In fact, the latter is also needed for a well-defined variational principle without having to impose extra-conditions on the normal derivatives of the metric on the boundary. The Euclidean path integral approach further needs the inclusion of a normalization term in order to obtain the right value for a Minkowski background. A more fundamental and direct derivation of the entropy of a black hole was introduced in \cite{Wald1993}, where the entropy is associated to the Noether charge of the horizon Killing field in a diffeomorphism invariant theory. The advantage is that no boundary terms nor normalization terms are needed and it is applicable to any theory with diffeomorphisms.

In this work we compute the entropy of a black hole using Wald's Noether charge method \cite{Wald1990,Wald1993,Wald1994,Wald1995} in a simpler formulation of GR based on non-metricity, the Symmetric Teleparallel Equivalent of GR (STEGR). For more details we refer the reader to the studies \cite{BeltranJimenez:2017tkd,Heisenberg18,Heisenberg_2019,jimenez2019,Heisenberg2020,HeisenbergL18}. In Wald's method for calculating black hole entropy one starts from diffeomorphism invariance, which leads to a current and a charge. Considering diffeomorphisms induced by a Killing vector field, the entropy is then defined as the integral of the charge at the bifurcation surface, or as argued after this definition of entropy, for a general cross section of the horizon. Furthermore the first law of black hole thermodynamics is derived. In the original papers, Wald's method was only applied to the standard formulation of GR equipped with the standard Levi-Civita connection. Here, we extend Wald's method for calculating the entropy to theories of gravity which involve an independent arbitrary connection. We mainly focus on the case of connections with zero curvature, zero torsion, but non-zero non-metricity, such as STEGR.

\vspace{0.5cm}

The remainder of the paper is organised as follows. In Section \ref{sec:preliminaries} we introduce some basic preliminary notions necessary for the rest of the work. In Section \ref{sec:wald} we extend Wald's method for calculating black hole entropy to theories with an independent torsion-free connection. In this section we will closely follow \cite{Wald1994}. In case ambiguities appear one can perform a reduction in phase space. Examples of reductions in phase space can be found in \cite{Wald1990}. An extension of the Noether charge calculation to theories with an independent connections which have non-zero torsion is presented in Appendix \ref{appendix:boundary_Palatini+torsion}. 
This can be applied to calculate entropy in Teleparallel Equivalent of GR (TEGR) for instance. In Section \ref{sec:STEGR} we calculate the entropy for STEGR. There we are free to add any exact form to the boundary, which vanishes exactly if the Lie derivative of the connection $\Gamma$ for the Killing vector field vanishes. This freedom is lost when performing a reduction of the naive phase space to the kinematically possible phase space. Taking into account this reduction of phase space, we end up with the same entropy for STEGR as the well-known classical entropy in GR from the Einstein-Hilbert action, as calculated in \cite{Wald1994}. Furthermore, we obtain a valid first law of black hole thermodynamics and an explicit expression for the energy appearing in the first law in STEGR. Fixing a gauge (namely choosing the coincident gauge), we see that the energy in STEGR in an asymptotically flat vacuum spacetime is the same as in GR, namely equal to the ADM mass. Finally, we summarize the main results in Section \ref{sec:summary}, where we also suggest topics for further research in this area.  

\vspace{0.5cm}

\section{Preliminaries: basic notions and conventions} \label{sec:preliminaries}
\subsection{Gauge theories of gravity: metric and independent connection}
In the literature the connection between Poincar\'e gauge theories and GR has been investigated since long. Choosing different gauge geometries, one can obtain different descriptions of GR. See for example the Geometrical Trinity in \cite{jimenez2019}. A standard way of formulating gauge descriptions of gravity uses a principal bundle setting, but the existence of a soldering form allows one to convert the bundle objects into objects on the spacetime. In this work we will not refer to bundle theory, but only consider all the objects as objects on the spacetime. We will use the abstract index notation introduced by Wald in \cite{Wald1990}; the tensor $T^{ab}{}_c$ is for example a $(2,1)$-tensor. Important to note is that one can always write all equations in abstract index notation, but this does not mean that there is coordinate invariance: for example the connection $\Gamma^a_{bc}$ is not a tensor in general. The notations and conventions used in this work are mainly the ones from \cite{Wald1994}, \cite{jimenez2019} and \cite{Heisenberg2020}. Following Wald's paper we put $k$-forms in bold: $\mathbf{\Theta}$. As we will be calculating entropy arising from a charge due to diffeomorphism invariance, let us briefly explain what diffeomorphism invariance means. A theory defined on a manifold $M$ is said to be diffeomorphic invariant if it is invariant under diffeomorphisms on $M$. If a theory is invariant under coordinate transformations of $M$, and does not have any fixed background fields, that is, all fields involved are dynamical, then a theory will be diffeomorphic invariant; the other way around can be more subtle, see for example \cite{pooley2015}.

\vspace{0.5cm}

On our spacetime $M$ we will consider the metric field $g$ and an independent connection $\Gamma^{a}_{bc}$ defining the covariant derivative $\nabla$. In the standard formulation of GR, the connection is uniquely determined by $g$, namely the Levi-Civita connection. Thus, the gravitational sector in our setting is described by a manifold $M$ and the independent fields $g$ and $\Gamma$, the metric and connection fields. If there is matter around, we should incorporate the coupling of this gravitational sector to matter, this is investigated for example in \cite{Heisenberg2020, Heisenberg_2019,jimenez2019}. However, here we will be interested in black hole entropy, and thus in solutions of gravity in vacuum, in the absence of matter fields. For the sake of clarity, let us distinguish the different covariant derivatives and derived quantities associated with the metric $g$ and connection $\Gamma$; the Levi-Civita connection associated with the metric will be denoted by $\mathcal{D}$, whereas the general connection associated to $\Gamma$ will be denoted by $\nabla$. The non-metricity tensor is defined by $Q_{abc} = \nabla_a g_{bc}$ and the torsion of the connection $\Gamma$ is defined by ${T^a}_{bc}=2{\Gamma^a}_{[bc]}$, see Tabel \ref{tabel:objectsgandgamma} for more notational conventions. We will use the convention where the covariant derivative acts on vectors and covectors through the connection symbols $\Gamma^a{}_{bc}$ by \begin{equation}
\begin{aligned}
\nabla_{b} V^{a} &=\partial_b V^{a}+\Gamma_{bc}^{a} V^{c}, \\
\nabla_{b} V_{a} &=\partial_b V_{a}-\Gamma_{b a}^{c} V_{c}.
\end{aligned}
\end{equation} A general connection can be decomposed into a Levi-Civita, a torsion and a non-metricity part: \begin{equation} \label{eq:connectiondecomp}
\Gamma_{b c}^{a}=\{ \substack{a\\bc} \}+K^a{}_{bc} + L^a{}_{bc},
\end{equation} where \begin{equation} \label{eq:relateLQ}
K^{a}{}_{b c}=\frac{1}{2} T^a{}_{bc}+T_{(b}{}^a{}_{c)}, \quad L^a{} _{bc}=\frac{1}{2} Q^a{}_{bc}-Q_{(b}{}^a{}_{c)}
\end{equation} are called the contortion and disformation tensors, arising respectively from the torsion and non-metricity. Here $\{ \substack{a\\bc} \}$ denotes the Levi-Civita Christoffel symbols. The curvature of the connection is then given by the Riemann tensor
\begin{equation}
    {R^a}_{bcd}=2\partial_{[c}{\Gamma^a}_{d]b}+2{\Gamma^a}_{[c|e|}{\Gamma^e}_{d]b}\ .
\end{equation}
GR, teleparallel and symmetric teleparallel theories are then obtained by requiring the vanishing of torsion and non-metricity, vanishing curvature and non-metricity, or vanishing curvature and torsion, respectively. The corresponding actions are then given by Lagrangians of either only Riemann, torsion, or non-metric tensors, while requiring the other two objects to vanish \cite{jimenez2019}. 
\\

\begin{table}[h]
\centering
\caption{Notation for objects related to the covariant derivatives associated to the metric $g$ and general connection $\Gamma$.}
\label{tabel:objectsgandgamma}
\begin{tabular}{lll}
  & $g$ & $\Gamma$ \\
  \toprule
 Covariant derivative  & $\mathcal{D}$ & $\nabla$ \\ 
 Curvature & $\mathcal{R}(g)$ & $R(\Gamma)$ (or simply $R$) \\ 
 torsion  & none & $T$ \\ non-metricity & none & $Q = \nabla g$ \\
\end{tabular}
\end{table}

\vspace{0.5cm}

\subsection{STEGR formulation of gravity}

The main theory which we will be interested in is STEGR, an equivalent of GR, where the connection is restricted to have zero curvature and zero torsion. From the decomposition of the connection, as given in \eqref{eq:connectiondecomp}, we see that the objects of interest are then the metric $g$ and the non-metricity tensor $Q$. Further details about such a construction for GR and beyond can be found for example in \cite{jimenez2019}. The most general even-parity second-order quadratic form of the non-metricity is given by \cite{BeltranJimenez:2017tkd,jimenez2019} \begin{equation}
\mathbb{Q}=\frac{c_{1}}{4} Q_{a b c} Q^{a b c}-\frac{c_{2}}{2} Q_{a b c} Q^{b a c}-\frac{c_{3}}{4} Q_{a} Q^{a}+\left(c_{4}-1\right) \tilde{Q}_{a} \tilde{Q}^{a}+\frac{c_{5}}{2} Q_{a} \tilde{Q}^{a},
\end{equation} where $Q^a = Q^{ab}{}_{b}$ and $\widetilde{Q}^a = Q_{b}{ }^{ b a}$. In STEGR we consider the quadratic form which fixes $c_i = 1$ for all $i$. For this quadratic form we have the following relation for a torsion-free connection: \begin{equation}
R=\mathcal{R}(g)+\mathbb{Q}+\mathcal{D}_{a}\left(Q^{a}-\widetilde{Q}^{a}\right),
\end{equation} so that we can relate actions with $\mathbb{Q}$ to the Einstein-Hilbert action which has $\mathcal{R}(g)$ as Lagrangian. The STEGR action is given by \begin{equation} \label{eq:STEGRpalatini}
\mathcal{S}_{\mathbb{Q}}=\int \mathrm{d}^{4} x\left[ -\frac{1}{2} \sqrt{-g} \mathbb{Q}+\lambda_{Ra}{}^{bcd} R^{a}{}_{bcd}+\lambda_{Ta}{}^{bc} T^{a}{}_{bc}\right],
\end{equation} where $\lambda_R$ and $\lambda_T$ are Lagrange multipliers setting the zero curvature and zero torsion constraints on the connection. Requiring the connection to be torsion- and curvature-free yields a connection of the following form \cite{jimenez2019}:
\begin{equation}
    \Gamma_{\mu \nu}^{\alpha}=\frac{\partial x^{\alpha}}{\partial \xi^{\lambda}} \partial_{\mu} \partial_{\nu} \xi^{\lambda}.
\end{equation}
The STEGR action can then be formulated in terms of the fields $g$ and $\xi$ by \begin{equation} \label{eq:STEGRnonpalatini}
    S = -\int \mathrm{d}^n x \frac{1}{2} \sqrt{-g} \text{ } \mathbb{Q}(g,\xi),
\end{equation} which is then equal to the Einstein-Hilbert action $S_{EH} = \int \mathrm{d}^n x \frac{1}{2}\sqrt{-g} \mathcal{R}(g)$ up to boundary terms. Making the choice $\xi^{\alpha} = x^{\alpha}$, where $x$ denote the coordinates, makes the connection $\Gamma$ vanish; this gauge choice is called the coincident gauge \cite{BeltranJimenez:2017tkd,jimenez2019}. The choice of STEGR in this gauge is also called Coincident General Relativity (CGR), in which case the action from STEGR (as given in \eqref{eq:STEGRnonpalatini}) is given by \begin{equation}
    \label{eq:CGR} S = \int \mathrm{d}^n x \frac{1}{2} \sqrt{-g} g^{cd} \big{(} \{ \substack{a \\bc} \} \{ \substack{b\\ad} \} - \{ \substack{a \\ba} \} \{ \substack{b \\cd} \}\big{)}.
\end{equation} An important point is that the latter action is only diffeomorphism invariant up to a boundary term \cite{jimenez2019}. As the goal of this report is calculating entropies using Wald's method where entropy arises from diffeomorphism invariance, we will not consider the gauged fixed version as in CGR but work in STEGR.

\vspace{0.5cm}

\section{Black hole entropy from Noether charge in Palatini formalism with torsion-free connection} \label{sec:wald} 
\subsection{Form of the symplectic potential and symplectic currents} \label{sec:Palatini_BT_sympl}
We denote by $\mathbf{L}$ the Lagrangian $n$-form which we integrate over in the action: \begin{equation} S = \int \mathbf{L}, \end{equation} thus $\mathbf{L} = \mathbf{L}(\phi)$ depends smoothly on the fields $\phi$ on $M$ which we consider, and $\mathbf{L}(\phi(x))$ is a $n$-form at $T_x M$ for every $x \in M$ \footnotetext{In standard formulation of GR for example the field which is being considered is the metric $g$ on $M$.}. We will consider variations $\delta \phi$ of our fields, where a variation is defined in the following way. Let $\phi(\lambda)$ be a one parameter family of fields on $M$, then the variation $\delta \phi$ is defined by \begin{equation}
    \delta \phi = \frac{\partial \phi(\lambda;x)}{\lambda} |_{\lambda = 0}.
\end{equation}  In the same way the variation of $L$ is defined as \begin{equation}
     \delta \mathbf{L} = \frac{\partial}{\partial \lambda} \mathbf{L} |_{\lambda = 0}. 
\end{equation}  Starting off in the same way as in Wald's paper, we vary the Lagrangian $\mathbf{L}$ in the action and separate the variation into the equations of motion $\mathbf{E}$ and a boundary term, leading to the definition of the boundary term $\mathbf{\Theta}$. 
\begin{definition}
Let $\mathbf{L}[\phi]$ denote the Lagrangian $n$-form, then the variation $\delta \mathbf{L}$ can be separated into equations of motion and a boundary:
\begin{equation} \label{eq:defTheta} \delta \mathbf{L} = \mathbf{E} \delta \phi + \mathrm{d} \mathbf{\Theta}(\phi, \delta \phi). \end{equation}  This defines the boundary term $\mathbf{\Theta}(\phi,\delta \phi)$ and the equations of motion $\mathbf{E}(\phi)$. A proof of this separation for theories with only the Levi-Civita connection can be found in \cite{Wald1994}, for the case of theories of gravity with independent connections see Lemma \ref{lemma:palatiniboundary} below. 
\end{definition}
\begin{remark} \label{remark:Y_ambiguity}
Notice that ambiguities in the definition of $\mathbf{\Theta}$ might arise, \begin{equation} \label{eq:waldambigtheta}
    \mathbf{\Theta} \rightarrow \mathbf{\Theta} + \delta \bm{\mu} + \mathrm{d}\mathbf{Y}(\phi,\delta \phi),
\end{equation}  obtained from the addition of an exact $n$-form $\mathrm{d}\bm{\mu}$ to $\mathbf{L}$, $\mathbf{L} \rightarrow \mathbf{L}+\mathrm{d}\bm{\mu}$, which leaves the equation of motion unaffected (and thus should not affect the dynamical content of the theory), and the freedom to add $\mathrm{d}\mathbf{Y}(\phi,\delta \phi)$ arises directly from the defining equation of $\mathbf{\Theta}$, namely \eqref{eq:defTheta}, see also p. 9 of \cite{Wald1990}. The form $\bm{Y}$ must be linear in the variations $\delta \phi$. In a theory involving only the Levi-Civita covariant derivative the entropy will be invariant under these ambiguities in the boundary term $\bm{\Theta}$, as proved in Proposition 4.1 and Theorem 6.1 in \cite{Wald1994}. However, as we will see, for theories with an independent connection one must carefully investigate these ambiguities; in the case of STEGR for example, the $\bm{Y}$-ambiguity only vanishes due to a reduction of the phase space. 
\end{remark}

Let us now prove \eqref{eq:defTheta}, from which we will also obtain a useful expression for the boundary term. If we want to refrain from fixing the connection the most appropriate approach is the Palatini formalism, where we work with an independent connection and a dynamical metric. We will only consider the case where the independent connection is torsion-free, with gravity action given by
\begin{equation} \label{eq:Palatiniaction1}
 \begin{split}
    S=\int d^4x \mathcal{L}\bigg( g_{ab}, &{R^a}_{bcd}, \nabla_{a_1}R^{a}_{bcd},\text{ } \ldots, \nabla_{(a_1}\ldots \nabla_{a_k)}R^{a}_{bcd},  Q_{abc}, \\ & \nabla_{a_1} Q_{abc},\text{ } \ldots, \nabla_{(a_1}\ldots \nabla_{a_l)}Q_{abc}, \Psi,\lambda\bigg) ,
\end{split}
\end{equation}
with the Lagrangian $\mathcal{L}$ an arbitrary density of weight $-1$, depending on the metric, connection, Riemann tensor and non-metric tensor, as well as extra dynamical fields $\Psi$ (indices of $\Psi$ are suppressed, but kept arbitrary), for example matter fields, and Lagrange multipliers $\lambda$ (this means that $\mathcal{L}$ depends only linearly on the $\lambda$'s, and does not contain derivatives of these). We also allow the Lagrangian density to depend on symmetrized $\nabla$-derivatives of the matter fields $\Psi$, but for notational simplicity we just denote this by $\Psi$ in the density $\mathcal{L}$ above. Note that derivatives of the metric are given by the non-metricity, and that in general for any linear connection terms as $\nabla_{a_1} \ldots \nabla_{a_k}$ can expressed in terms of the symmetrized derivative $\nabla_{(a_1} \ldots \nabla_{a_k)}$ plus curvature terms, torsion terms and lower order derivatives \cite{Wald1990,carroll_2019}, which is why we only incorporate symmetrized covariant derivatives in $\mathcal{L}$.

\vspace{0.5cm}

Now let us vary this action to obtain the equations of motion and the boundary term. We proceed analogously as in Wald's approach \cite{Wald1994}. The Lagrangian is given by an $n$-form (with $n$ the spacetime dimension), written as \textbf{L}, such that for a given Lagrangian density $\mathcal{L}$ we have
\begin{equation}
    \textbf{L}=\varepsilon\mathcal{L}
\end{equation}
where $\varepsilon$ is the totally antisymmetric symbol. In Wald's original paper \cite{Wald1994} the factor $\sqrt{-g}$ is also included in the totally antisymmetric symbol, i.e. $\bm{\epsilon}_{a_1...a_n}=\sqrt{-g}\varepsilon_{a_1...a_n}$, but since we work now with a possibly not metric-compatible connection we rather include the determinant in the Lagrangian and work with the Lagrangian density. As shown in the following lemma, we can split the variation of $\mathbf{L}$ into equations of motion and boundary terms. 

\begin{lemma} \label{lemma:palatiniboundary}
We can write a general variation of the Lagrangian density $\mathcal{L}$, defined in \eqref{eq:Palatiniaction1}, as
\begin{equation} \begin{split}
    {\delta\mathcal{L}}&=(A^{ab}_g-\nabla_c E_Q^{cab})\delta g_{ab}\\ &+(B_{\Gamma a}^{bc}-2\nabla_d {E_{Ra}}^{cdb}-2E_Q^{bc}{}_a)\delta {\Gamma^a}_{bc}\\
 &+E_\Psi\delta\Psi + C_\lambda \delta\lambda\\
&+\nabla_d(2{E_{Ra}}^{cdb}\delta{\Gamma^a}_{bc}+\sqrt{-g}{E_Q}^{dab}\delta g_{ab})+\nabla_d{\tilde\Theta}^d\  .
\end{split} \end{equation}
The first two lines give the equations of motion for the metric and the connection, while the third line are the equations of motion for the extra fields $\Psi$ and the constraints. The fourth line gives the desired boundary term. Since the $E_Q^{dab}\delta g_{ab}$ part of the boundary term will only contain variations of the metric, it will ultimately not contribute to the entropy, and we thus absorb it in $\tilde\Theta$. The boundary term thus is
\begin{equation} 
    \Theta^d=2{E_{Ra}}^{cdb}\delta{\Gamma^a}_{bc}+\tilde\Theta^d\ ,
\end{equation}
where the contribution $\tilde{\boldsymbol{\Theta}}$ contains only terms like $\delta\nabla_{a_1}...\nabla_{a_i}R_{abcd}$, $\delta \nabla_{a_1} Q_{abc}$ and $\delta \Psi$ etc, e.g. it contains only variations but not derivatives of variations. The functions $E_i$ are the equations of motion for the curvature, non-metricity tensor, and extra fields $\Psi$, as if they were viewed as independent fields. For example
\begin{equation}
    E_R^{abcd}=\frac{\partial\mathcal{L}}{\partial R_{abcd}}-\nabla_{a_1}\frac{\partial\mathcal{L}}{\partial \nabla_{a_1}R_{abcd}}+ \ldots+(-1)^{m} \nabla_{(a_{1}} \ldots \nabla_{a_{m})} \frac{\partial L}{\partial \nabla_{(a_{1}} \ldots \nabla_{a_{m})} R_{a b c d}} .
\end{equation}
\end{lemma}
\begin{proof}
Varying $\textbf{L} = \varepsilon \mathcal{L}$ gives
\begin{align}
\nonumber    \delta \textbf{L}&=\varepsilon\left(\frac{\partial \mathcal{L}}{\partial g_{ab}}\delta g_{ab}+\frac{\partial\mathcal{L}}{\partial R_{abcd}}\delta R_{abcd}+\frac{\partial\mathcal{L}}{\partial\nabla_{a_1}R_{abcd}}\delta\nabla_{a_1}R_{abcd}+...\right.\\
\nonumber\\
\nonumber &+\frac{\partial\mathcal{L}}{\partial Q_{abc}}\delta Q_{abc}+\frac{\partial\mathcal{L}}{\partial\nabla_{a_1}Q_{abc}}\delta\nabla_{a_1}Q_{abc}+...\\
\nonumber\\
 &\left.+\frac{\partial\mathcal{L}}{\partial \Psi}\delta \Psi+\frac{\partial\mathcal{L}}{\partial\nabla_{a_1}\Psi}\delta\nabla_{a_1}\Psi+...+\frac{\partial\mathcal{L}}{\partial \lambda}\delta\lambda\right)\ ,
\end{align}
where $...$ stands for variations with respect to higher-order derivatives of $R,Q$ or $\Psi$. The variation of $\sqrt{-g}$ is included in the first term. In order to treat the terms containing derivatives of Riemann and non-metricity tensors, we compute for any tensor $A$
\begin{align}
\nonumber    \frac{\partial\mathcal{L}}{\partial \nabla_{(a_1}...\nabla_{a_i)}A}\delta \nabla_{(a_1}...\nabla_{a_i)} A&=\frac{\partial\mathcal{L}}{\partial \nabla_{(a_1}...\nabla_{a_i)}A} \nabla_{a_1}\delta \nabla_{a_2}...\nabla_{a_i}A+\\
\nonumber    &+ \text{terms with }\delta \Gamma\\
\nonumber    &=\nabla_{a_1}\left(\frac{\partial\mathcal{L}}{\partial \nabla_{(a_1}...\nabla_{a_i)}A}\delta\nabla_{a_2}...\nabla_{a_i}A\right)+\\
\nonumber    &+ \text{terms with }\delta \Gamma\\
    &-\nabla_{a_1}\left(\frac{\partial\mathcal{L}}{\partial \nabla_{(a_1}...\nabla_{a_i)}A}\right)\delta \nabla_{a_2}...\nabla_{a_i}A\ ,
\end{align}
where the terms with $\delta\Gamma$ arise from the variation of the first covariant derivative $\nabla_{a_1}$ that we pulled out \footnotetext{In Wald's original procedure the Levi-Cevita connection is used, such that the variation of $\nabla$ gives a term containing $\nabla\delta g_{ab}$ \cite{Wald1994}. Integrating this by parts gives then a contribution to the metric equations of motion, and a boundary term containing $\delta g_{ab}$.}. Each of the terms containing $\delta{\Gamma^a}_{bc}$ will contribute to the connection equation of motion. The first term in the last expression will contribute to the boundary term. The term in the last line is then of lower differential order, and can be added to the contribution $\frac{\partial\mathcal{L}}{\partial \nabla_{(a_1}...\nabla_{a_{i-1})}A}\delta \nabla_{(a_1}...\nabla_{a_{i-1})} A$ coming from $\delta\textbf{L}$. Performing this procedure for all derivative terms in $\delta\textbf{L}$ and collecting all terms of each of the variations one finds
\begin{equation}
\label{eq:Palatinivariation1}
    \delta\textbf{L}=A^{ab}_g \delta g_{ab}+B_{\Gamma a}^{bc} \delta \Gamma^a_{bc}+E_R^{abcd}\delta R_{abcd}+E_Q^{abc}\delta Q_{abc}+ E_\Psi\delta\Psi + C_\lambda \delta\lambda+d\tilde{\boldsymbol{\Theta}} ,
\end{equation} where the contribution $\tilde{\boldsymbol{\Theta}}$ contains only terms like $\delta\nabla_{a_1}...\nabla_{a_i}R_{abcd}$, $\delta \nabla_{a_1} Q_{abc}$ and $\delta \Psi$ etc, e.g. it contains only variations but not derivatives of variations. The functions $E_i$ are then the equations of motion for the curvature, non-metricity tensor, and extra fields $\Psi$, as if they were viewed as independent fields. For example
\begin{equation}
    E_R^{abcd}=\frac{\partial\mathcal{L}}{\partial R_{abcd}}-\nabla_{a_1}\frac{\partial\mathcal{L}}{\partial \nabla_{a_1}R_{abcd}}+ \ldots+(-1)^{m} \nabla_{(a_{1}} \ldots \nabla_{a_{m})} \frac{\partial L}{\partial \nabla_{(a_{1}} \ldots \nabla_{a_{m})} R_{a b c d}} .
\end{equation}
The functions $C_\lambda$ are the constraints coming from the Lagrange multipliers. Note that $A_g$ also contains the terms coming from variations of factors containing the determinant of the metric $\sqrt{-g}$, which we have not made explicit in \eqref{eq:Palatiniaction1}. \\
Finally, we have that the variations appearing in \eqref{eq:Palatinivariation1} can be written as variations of the metric and connection as (recall that we work with a torsion-free connection here)
\begin{equation} \label{eq:variation_of_R_Q} \begin{split}
    \delta {R^a}_{bcd}&=2\nabla_{[c}\delta{\Gamma^a}_{d]b}\\
    \delta Q_{abc}&= \nabla_a\delta g_{bc}-2\delta {\Gamma^{d}}_{a(b}g_{c)d}\ .
\end{split} \end{equation}  The $\delta\Gamma$ terms from $\delta Q$ will only contribute to the connection equation of motion. The contributions to the boundary term come from $\delta R$ and the metric variation in $\delta Q$. Hence
\begin{equation} \label{eq:varL5.11} \begin{split}
    {\delta\mathcal{L}} &=(A^{ab}_g-\nabla_c E_Q^{cab})\delta g_{ab}\\ &+(B_{\Gamma a}^{bc}-2\nabla_d {E_{Ra}}^{cdb}-2E_Q^{bc}{}_a)\delta {\Gamma^a}_{bc}\\
 &+E_\Psi\delta\Psi + C_\lambda \delta\lambda\\
&+\nabla_d(2{E_{Ra}}^{cdb}\delta{\Gamma^a}_{bc}+E_Q^{dab}\delta g_{ab})+\nabla_a{\tilde\Theta}^a\ ,
\end{split} \end{equation}
where we used the symmetries of the $E_i$'s, inherited from the objects through which they are defined. The first and second line give the equations of motion for the metric and connection, while the third line are the equations of motion for the extra fields $\Psi$ and the constraints. The fourth line gives the desired boundary term. Since the $E_Q$ part of the boundary term only contains a variation of the metric, it will ultimately not contribute to the entropy (see Remark \ref{remark:importantboundary}), and we thus absorb it in $\tilde\Theta$. The boundary term can thus be written as
\begin{equation}
    \Theta^a=2{E_{Rd}}^{cab}\delta{\Gamma^d}_{bc}+\tilde\Theta^a\ ,
\end{equation}
or in form language
\begin{equation}
    (\boldsymbol{\Theta})_{a_2 \ldots a_n}=\varepsilon_{a a_2 \ldots a_n} 2{E_{Rd}}^{cab}\delta{\Gamma^d}_{bc}+(\tilde{\boldsymbol{\Theta}})_{a_2 \ldots a_n}\ .
\end{equation}
\end{proof}
\begin{remark}
In \cite{Wald1994}, other connections than the Levi-Civita connections are initially also considered, but they are then chosen to be fixed background connections. Here we are doing something different: we are allowing for an arbitrary independent connection to exist, different from the Levi-Civita connection.
\end{remark}
\begin{remark}
The covariant derivative $\left(\nabla_{c} A\right)^{a_{1} \ldots a_{r}}{}_{ b_{1} \ldots b_{s}}$ of a tensor density $A$ is equal to the covariant derivative of the density as if it was a tensor, plus a term $+W\Gamma_{d c}^{d} T^{a_{1} \cdots a_{r}}{}_{b_{1} \ldots b_{s}}$ with $W$ the weight of the density. Thus for a vector density $A^a$ of weight $-1$ and a torsion-free connection we have \begin{equation}
    \nabla_a A^a = \partial_a A^a + \Gamma^a_{ab} A^b - \Gamma^b_{ab} A^a = \partial_a A^a.
\end{equation} This is why the terms $\nabla_d ( \ldots )$ in \eqref{eq:varL5.11} are pure boundary terms for a torsion-free connection: the variation of the action $S$ can be written here as \begin{equation}
    \delta S \sim \int \mathrm{d}^4 x \nabla_d \Theta^d = \int \mathrm{d}^4 x \partial_d(\Theta^d) = \int \mathrm{d}\bm{\Theta},
\end{equation} where $\bm{\Theta}$ is the Hodge dual of the tensor $A^a = \sqrt{-g}\Theta^a$. If the connection is not torsion-free, the above calculations for splitting the variation into the equations of motion and boundary terms have to be altered, see also Appendix \ref{appendix:boundary_Palatini+torsion}.  
\end{remark}

\vspace{0.5cm}

\begin{definition} \label{def:omega}
Let $\phi(\lambda_1,\lambda_2)$ be a two-parameter family of fields, yielding variations $\delta_1 \phi$ and $\delta_2 \phi$, and $(n-1)$-forms $\mathbf{\Theta}(\phi,\delta_1 \phi)$ and $\mathbf{\Theta}(\phi,\delta_2 \phi)$. The $(n-1)$-form $\bm{\omega}(\phi,\delta_1 \phi,\delta_2 \phi)$, defined by \begin{equation}
    \bm{\omega}(\phi,\delta_1 \phi,\delta_2 \phi) = \delta_1 \mathbf{\Theta}(\phi,\delta_2 \phi) - \delta_2 \mathbf{\Theta}(\phi, \delta_1 \phi),
\end{equation} is called the symplectic current. The presymplectic form relative to a Cauchy surface $C$ is defined by \begin{equation} \label{eq:presympl}\Omega(\phi,\delta_1 \phi, \delta_2 \phi) = \int_C \bm{\omega}(\phi,\delta_1 \phi, \delta_2 \phi). \end{equation} 
\end{definition}

\vspace{0.5cm}

\begin{remark}

If we first vary $\mathbf{L}$ with respect to $\delta_2 \phi$, and then with respect to $\delta_1 \phi$, we get \begin{equation}
    \delta_1 \delta_2 \mathbf{L} = \frac{\partial^2 \mathbf{L}}{\partial \lambda_1 \partial \lambda_2} = (\delta_1 \mathbf{E} )(\delta_2 \phi) + \mathbf{E} (\delta_1 \delta_2 \phi) + \mathrm{d}\delta_1 \bm{\Theta}(\phi,\delta_2 \phi).
\end{equation} We can do the same for $\delta_2 \delta_1 \mathbf{L}$, and by subtracting these two expressions, using the equality of mixed partial derivatives, we get the equation \begin{equation}
    (\delta_1 \mathbf{E})(\delta_2 \phi) - (\delta_2 \mathbf{E}) (\delta_1 \phi) + \mathrm{d} \bm{\omega} = 0,
\end{equation} which means that if $\delta_i \mathbf{E} = 0$ for $i=1,2$, then $\mathrm{d}\bm{\omega} = 0$. If the condition $\delta_i \mathbf{E} = 0$ is satisfied, the variation $\delta_i \phi$ is said to satisfy the linear equations of motion. For an infinitesimal local symmetry $\hat{\delta} \phi$ at $\phi$, where $\phi$ satisfies the equation of motion, the linearized equations of motion are satisfied, as proved in \cite{Wald1990}. So if the linearized equations of motion are satisfied by the infinitesimal variations we have by Stokes theorem that for two Cauchy surfaces $C,C'$ the difference $\int_C \bm{\omega} - \int_{C'} \bm{\omega}$ depends on an integral $\int_S \bm{\omega}$, where $S$ is a timelike surface near spatial infinity, which will typically vanish when the integral of $\bm{\omega}$ over a Cauchy surface $C$ is required to converge in order to have a well-defined presymplectic form in \eqref{eq:presympl}, as described in \cite{Wald1994}. If not so, one could strengthen the asymptotic conditions imposed upon the dynamical fields, so that the symplectic form is independent of the chosen Cauchy surface. 
\end{remark}
\begin{remark}
Note that the presymplectic form is independent under the $\bm{\mu}$-ambiguity in $\bm{\Theta}$, as defined in \eqref{eq:waldambigtheta}, that is under $\mathbf{\Theta} \rightarrow \mathbf{\Theta} + \delta \bm{\mu}$ the presymplectic form $\omega$ is unchanged. The $\bm{Y}$-ambiguity can be trickier, as mentioned also in the following remark.
\end{remark}
\begin{remark} \label{remark:Y_ambig_symplecticform_Palatini}
For terms in $\bm{Y}$ linear in $\delta g$ we can use the same argument as given on p. 9-10 in \cite{Wald1994} to argue that such terms will not contribute to $\Omega$. A symplectic form can be obtained from the presymplectic form $\Omega$ by taking an appropriate quotient of phase space, as described in Wald's paper \cite{Wald1990}. It is hoped for that because of this reduction the $\delta \Gamma$-linear terms in $\bm{Y}$-ambiguity in \eqref{eq:waldambigtheta} do not lead to ambiguities in the symplectic form or other quantities such as the entropy which we will define below. This will be the case in STEGR.
\end{remark}

\vspace{0.5cm}

\subsection{Form of the Noether charge} \label{sec:Form_Noethercharge}

In looking at variations of fields, we will be interested in variations generated by a vector field $V$. 
\begin{definition}
Let $\mathbf{L}$ be a Lagrangian with boundary term $\bm{\Theta}$ and let $V$ be a vector field. The Noether current $\mathbf{J}$ is then defined by \begin{equation} \label{eq:defcurrentJ} \mathbf{J} = \bm{\Theta}(\phi,\mathcal{L}_V \phi) - i_{V} \mathbf{L} - \bm{\alpha}_{V}, \end{equation} where the last term $\bm{\alpha}_{V}$ arises if a theory is diffeomorphism invariant up to a boundary term $\mathrm{d} \bm{\alpha}_{V}$, that is, if $\hat{\delta} \mathbf{L} =  \mathbf{L}(\mathcal{L}_{V} \phi) = \mathcal{L}_{V} \mathbf{L}(\phi) + \mathrm{d}\bm{\alpha}_{V}$.
\end{definition}
\begin{remark}
In the above definition we have denoted variations induced by infinitesimal local symmetries by $\hat{\delta}$, from now on we will often omit the $\hat{\text{hat}}$-symbol when it is clear that variations are induced by infinitesimal local symmetries. The symbol $i_{V}$ acting on a tensor $A(\cdot,\cdot,\ldots,\cdot)$, where $\cdot$ denotes the interior product, defined by \[ i_{V} A(\cdot, \cdot, \ldots, \cdot) =  A(V, \cdot, \ldots, \cdot). \] 
\end{remark}
\begin{remark}
In the original approach diffeomorphism invariant theories were only considered, hence the term $\bm{\alpha}_{V}$ was not present. In the remaining of this work we will work with diffeomorphism invariant Lagrangians, but for completeness we will still allow for Lagrangians with diffeomorphism invariance up to a boundary term, as it is the case for CGR (gauge fixed STEGR).
\end{remark}

\vspace{0.5cm}

\begin{lemma} \label{lemma:J_closed}
The current $\mathbf{J}$ is closed when evaluated on the equations of motion, that is $\mathrm{d}\mathbf{J} = 0$ when evaluated on fields which obey the equations of motion. \end{lemma}
\begin{proof}
From the definition of $J$ above we have the following: \begin{equation*}
    \begin{split}
        \mathrm{d} \mathbf{J} &= \mathrm{d}\bm{\Theta} - \mathbf{d}(i_{V} \mathbf{L}) - \mathrm{d}\bm{\alpha}_{V} \\ &= \mathrm{d}\bm{\Theta} -\mathcal{L}_{V}  \mathbf{L} - \mathrm{d}\bm{\alpha}_{V} \\ &= -\hat{\delta} \mathbf{L} + \mathrm{d}\bm{\Theta} = - \mathbf{E} \hat{\delta} \phi,
    \end{split}
\end{equation*} where we used the general identity  \begin{equation}
    \label{eq:generalidentity} \mathcal{L}_{V} \bm{\Lambda} = i_{V} \mathrm{d}\bm{\Lambda} + \mathrm{d}(i_{V} \bm{\Lambda}),
\end{equation}  which holds for any differential form $\bm{\Lambda}$ and any vector field $V$, together with the fact that $\mathrm{d} \mathbf{L} = 0$ as $\mathbf{L}$ is an $n$-form on an $n$-dimensional manifold. 
\end{proof}

\begin{lemma} \label{lemma:noetherchargewald}
There exists an $(n-2)$-form $\mathbf{Q}_N$, also called the Noether charge, such that $\mathbf{J} = \mathrm{d}\mathbf{Q}_N$. 
\end{lemma}
\begin{proof}
As $J$ is closed on all fields which obey the equations of motion, we conclude by the same argument as in the original paper that there exists a $(n-2)$-form $Q$ such that $\mathbf{J} = \mathrm{d} \mathbf{Q}_N$ \cite{Wald1994}. An explicit algorithm can be found in Lemma 1 of \cite{Wald1990b}; this algorithm is also explained together with an example (namely the STEGR Noether charge) in Appendix \ref{appendix:Noethercharge_STEGR}.
\end{proof} 
\begin{remark} \label{remark:importantboundary}
From the explicit algorithm in \cite{Wald1990b} of finding the Noether charge $\mathbf{Q}_N$ from the current $\mathbf{J}$, one sees that when going from $\mathbf{J}$ to $\mathbf{Q}_N$, every term having derivatives $\nabla$ of $V$ in $\mathbf{J}$ goes to one order of derivative lower in $\mathbf{Q}_N$. We will calculate the entropy from the integral of the Noether charge over the bifurcation surface, and we will choose the vector field $V$ to be the Killing vector field, which vanishes on the bifurcation surface. Thus we see that at least second-order derivatives of $V$ in the boundary term $\bm{\Theta}$ are needed to have a contribution to the entropy. Such contributions will be called `important' contributions to the boundary term.
\end{remark}

Let us now look at the form of the Noether charge for theories with an arbitrary independent connection.
\begin{lemma} \label{lemma:formNoerthercharge_dynconnection}
For a theory with an arbitrary independent connection, considering Lie derivative variations induced by a vector field $V$ the Noether charge has the form \begin{equation}
    \mathbf{Q}_N=\mathbf{X}^{cd}\nabla_{c}V_{d}+ \text{terms linear in $V$}
\end{equation}
with
\begin{equation}
    (\mathbf{X}^{cd})_{a_3...a_n}=\varepsilon_{aba_3...a_n}{E_{R}}^{dcab}\ ,
\end{equation} up to the $\bm{Y}$-ambiguity mentioned in Remark \ref{remark:Y_ambiguity}.
\end{lemma}
\begin{proof}
From the boundary term $\boldsymbol\Theta$ (see Lemma \ref{lemma:palatiniboundary} for an explicit expression) we obtain the Noether current associated to a diffeomorphism generated by a vector field $V$ by
\begin{equation}
    \mathbf{J}=\bm{\Theta}(\delta g_{ab}=\Lie g_{ab},\delta{\Gamma^a}_{bc}=\Lie {\Gamma^a}_{bc})-i_V \mathbf{L}\ ,
\end{equation}
with the Lie derivatives given by (note that \cite{Heisenberg2020} uses a different sign notation)
\begin{align}
    \Lie g_{ab}&=2g_{c(a}\nabla_{b)}V^c-(2T_{(ab)c}-Q_{cab})V^c\ ,\\
    \Lie {\Gamma^a}_{bc}&=\nabla_b\nabla_cV^a-{T^a}_{cd}\nabla_bV^d-({R^a}_{cbd}+\nabla_b {T^a}_{cd})V^d\ .
\end{align}
One can then see that all the terms in \textbf{J} will contain only single or no derivatives in $V$, except for the term containing $E_R\delta\Gamma$ (supressing indices). As the part of the boundary denoted by $\tilde{\bm{\Theta}}$ (see Lemma \ref{lemma:palatiniboundary}) only contains terms like $R_{abcd}, \delta \nabla_{a_1} R^{abcd}$, that is, variations but no derivatives of variations, this term will only yield terms linear in $V$ and $\nabla V$. This gives the current as
\begin{equation} \label{eq:J_symmetricPalatini}
    (\textbf{J})_{a_2 a_3 \ldots a_n}=\varepsilon_{a a_2 a_3 \ldots a_n}2{E_{Rd}}^{cab}\nabla_b\nabla_c V^d+\text{terms linear in $V$ and $\nabla V$}
\end{equation}
Using the same algorithm as in Lemma 1 of \cite{Wald1990b} one then obtains the Noether charge $\mathbf{Q}_N$ as (for a calculation see Appendix \ref{appendix:Noethercharge_STEGR})
\begin{equation}
    \mathbf{Q}_N=\mathbf{X}^{cd}\nabla_{c}V_{d}+ \text{terms linear in $V$}
\end{equation}
with
\begin{equation}
    (\mathbf{X}^{cd})_{a_3...a_n}=\varepsilon_{aba_3...a_n}{E_{R}}^{dcab}\ .
\end{equation} 
\end{proof}
\begin{remark}
\label{remark:Y_ambiguity_noethercharge_Palatini}
Note that the above boundary term is only determined up to the $\bm{Y}$-ambiguity (see Remark \ref{remark:Y_ambiguity}), that is at this point we are free to add any $(n-2)$-form $\bm{Y}$ linear in $\mathcal{L}_V g$ or $\mathcal{L}_V \Gamma$ to the Noether charge $\bm{Q}_N$. For the entropy and first law the terms linear in $\mathcal{L}_V g$ will not contribute as for the entropy $V$ will be Killing vector field so that $\mathcal{L}_V g = 0$ (see also p. 17 in \cite{Wald1994}). However, terms linear in $\mathcal{L}_V \Gamma$ might contribute to the entropy and to the first law, but in STEGR for example, by a reduction of phase space such terms can be set equal to zero. This reduction is described in Section \ref{sec:reduction_phasespace}. Furthermore, from the definition of the Noether charge there is freedom to add any exact form $\mathrm{d}\bm{Z}$ to the Noether charge. Following argumentation as in the proof of Theorem 6.1 in \cite{Wald1994}, to calculate the entropy we will integrate over a surface without boundary so that such an exact form does not contribute to the entropy; the same holds for the energy and angular momentum appearing in the first law, as defined in the next section.
\end{remark}

\vspace{0.5cm}

\subsection{Definition of energy and angular momentum from the symplectic form} 

Now let us look at variations of  the current $\mathbf{J}$, as defined in \eqref{eq:defcurrentJ}, induced by an arbitary variation $\delta \phi$ of fields $\phi$ satisfying the equation of motion, and let $V$ be some vector field on $M$:
\begin{equation}
\begin{aligned}
\delta \mathbf{J} &=\delta \boldsymbol{\Theta}\left(\phi, \mathcal{L}_{V} \phi\right)-i_V  \delta \mathbf{L} - \delta \bm{\alpha}_V \\
&=\delta \boldsymbol{\Theta}\left(\phi, \mathcal{L}_{V} \phi\right)-i_V  \mathrm{d} \boldsymbol{\Theta}(\phi, \delta \phi)  - \delta \bm{\alpha}_V\\
&=\delta \boldsymbol{\Theta}\left(\phi, \mathcal{L}_{V} \phi\right)-\mathcal{L}_{V} \boldsymbol{\Theta}(\phi, \delta \phi)+\mathrm{d}(i_V \boldsymbol{\Theta}(\phi, \delta \phi)) - \delta \bm{\alpha}_V
\end{aligned}
\end{equation}
where in line 1 we used the fact that $\mathbf{E} = 0$ on $\phi$ and in line 2 we used the general identity \eqref{eq:generalidentity}. Recalling the definition of the symplectic current $\bm{\omega}$ we see from the above equation that we have \begin{equation}  \label{eq:75}
    \bm{\omega}(\phi,\delta \phi, \mathcal{L}_{V} \phi) = \delta \mathbf{J} - \mathrm{d}(i_{V} \bm{\Theta}(\phi, \delta \phi) + \bm{\alpha}_{V}). 
\end{equation} Integrating this equation over a Cauchy surface, and choosing $V$ to be a timelike vector field we obtain Hamilton's equations of motion by \begin{equation}
   \begin{split}
        \label{eq:Hameq} \delta H &= \Omega(\phi,\delta \phi, \mathcal{L}_{V} \phi) \\ &= \int_C \bm{\omega}(\phi,\delta \phi, \mathcal{L}_{V} \phi) = \delta \int_C \mathbf{J} - \int_C \mathrm{d}(i_{V} \bm{\Theta} + \bm{\alpha}_{V}) \\ &= \delta \int_C \mathbf{J} - \int_{\infty} (i_{V} \bm{\Theta} + \bm{\alpha}_{V}), 
   \end{split}  
\end{equation} and thus if we assume that there exists a $(n-1)$-form $\mathbf{B}$ and an $n-2$-form $\mathbf{C}_{V}$ such that \begin{equation}
    \delta \int_{\infty}( i_{V} \mathbf{B} + \mathbf{C}_{V}) = \int_{\infty} (i_{V} \bm{\Theta} + \bm{\alpha}_{V}), 
\end{equation} we get that $H$ is given by \begin{equation}
    H = \int_C\mathbf{J} - \int_{\infty} (i_{V} \bm{B} + \bm{C}_{V}).
\end{equation} Evaluating $H$ on solutions we have, using $\mathbf{J} = \mathrm{d}\mathbf{Q}_N$, that \begin{equation}
    H = \int_{\infty} (\mathbf{Q}_N - i_{V} \mathbf{B} + \mathbf{C}_{V}). 
\end{equation}
The canonical energy $E$ can then be defined as
 \begin{equation} \label{eq:canenergy}
    E = \int_{\infty} (\mathbf{Q}_N[t] - i_{t} \mathbf{B} + \mathbf{C}_{t}), 
\end{equation} where $t^a$ is an aymptotic time translation vector field. Similarly, we may define a canonical angular momentum $\mathcal{J}$ in the following way. We let $\varphi^a$ be an asymptotic rotation, and we choose the surface at infinity such that $\varphi^a$ is everywhere tangent to the surface, in that case the pullback of $i_{\varphi} \bm{\Theta}$ has input vectors everywhere tangent to the surface we have chosen, which is $(n-2)$-dimensional, and as $\bm{\Theta}$ is an $(n-1)$-form, this pullback of $i_{\varphi} \bm{\Theta}$ to the surface must be zero. This allows us to define the canonical angular momentum $\mathcal{J}$ by \begin{equation}
    \label{eq:canangular} \mathcal{J} = - \int_{\infty} (\mathbf{Q}_N[\varphi]+ \mathbf{C}_{\varphi}).
\end{equation} Again we should note that these definitions might be ambiguous up to terms in the $\mathbf{Y}$-ambiguity linear in $\delta \Gamma$ (see \eqref{eq:waldambigtheta}). As we will see this ambiguity will vanish for STEGR when going to the reduced phase space, leading to a well-defined energy and angular momentum from \eqref{eq:canenergy} and \eqref{eq:canangular}.

\vspace{0.5cm} 

\subsection{Entropy and the first law of black hole thermodynamics} \label{sec:firstlawBH}
Using \eqref{eq:75} and the definitions of canonical energy and canonical angular momentum as defined in \eqref{eq:canenergy} and \eqref{eq:canangular} we can try to get a version of the first law of black hole thermodynamics. Let us consider the case where the vector field $V$ is a symmetry of the dynamical fields, that is $\mathcal{L}_{V} \phi = 0$, and $\delta \phi$ satisfies the linearised equations of motion (that is $\delta \bm{E} = 0$ where the variation is induced by $\delta \phi$). Then $\omega(\phi,\delta \phi, \mathcal{L}_{V} \phi) =0$, and thus we have, using $\bm{J} = \mathrm{d}\mathbf{Q}_N$ that \[  0 = \delta \mathrm{d}\mathbf{Q}_N - \mathrm{d}(i_{V} \bm{\Theta} + \bm{\alpha}_{V}).  \] Integrating this identity over a hypersurface $\Xi$ with boundary $\partial \Xi$, we get that \begin{equation} \label{eq:precursor1stlaw}
    \int_{\partial \Xi} \big{(} \delta \mathbf{Q}_N - i_{V} \bm{\Theta}(\phi,\delta \phi)-\bm{\alpha}_{V}\big{)} = 0.
\end{equation}  Now by choosing a specific form of hypersurface $\Xi$ and a Killing vector field $V$, we get a precursor of the first law. Suppose we have a stationary black hole with bifurcate Killing horizon \cite{Wald1994, Racz_1992}. Let us choose $\Xi$ to be an asymptotically flat hypersurface, which has two boundaries: on the one hand a surface at infinity, and on the other hand the bifurcation $(n-2)$-surface $\Sigma$ of the black hole as its interior boundary. For the boundary parts at infinity we can use the definitions of canonical energy and canonical angular momentum, see \eqref{eq:canenergy} and \eqref{eq:canangular}, whereas the boundary surface $\Sigma$ can lead to the definition of entropy. For the Killing vector field we choose the Killing vector field which vanishes on the bifurcation surface, with a normalization such that \begin{equation}
    V^a = t^a + \Omega_H^{(\mu)} \varphi^a_{(\mu)},
\end{equation} where $t^a$ is the stationary Killing vector field which has norm $1$ at infinity, and $\varphi^a_{(\mu)}$ are the asymptotic rotations in orthogonal planes, and $\Omega_H^{(\mu)}$ are the angular velocities of the horizon. In this way the equation above, \eqref{eq:precursor1stlaw}, becomes \begin{equation}
    \label{eq:1stlawQ} \delta \int_{\Sigma} (\mathbf{Q}_N[V]+\mathbf{C}_{V}) = \delta E -   \Omega_H^{(\mu)} \delta \mathcal{J}_{(\mu)}.
\end{equation} Thus, if for the left hand side of this equation we have an equality of the form \begin{equation} \label{eq:Qkentropy_notdiff}
    \delta \int_{\Sigma} (\mathbf{Q}_N[V]+\mathbf{C}_{V}) = \frac{\kappa}{2\pi} \delta S ,
\end{equation} for some $S$, we obtain the first law of black hole thermodynamics from \eqref{eq:1stlawQ}, that is, \begin{equation} \label{eq:firstlaw_stated}
\frac{\kappa}{2 \pi} \delta S=\delta {E}-\Omega_{H}^{(\mu)} \delta \mathcal{J}_{(\mu)}
\end{equation} for variations $\delta \phi$ satisfying the linearised equations of motion, and might define $S$ as the entropy, as discussed below. Here $\kappa$ denotes the surface gravity of the black hole. Recall that $\mathbf{C}_{V}$ arises if the Lagrangian $n$-form $\mathbf{L}$ was diffeomorphism invariant only up to a boundary term $\bm{\alpha}_{V}$. As we will only consider diffeomorphism invariant theories here, our definition of the entropy will be \begin{equation} \label{eq:Qkentropy}
    \delta \int_{\Sigma} \mathbf{Q}_N[V] = \frac{\kappa}{2\pi} \delta S ,
\end{equation} 

Now we turn to the case of a stationary black hole solution with bifurcate Killing horizon. Recall the expression for the Noether charge as given in Lemma \ref{lemma:formNoerthercharge_dynconnection}. Let $V$ denote the Killing vector field (that is, $\mathcal{L}_V g = 0$) which vanishes on the bifurcation $(n-2)$-surface $\Sigma$. Integrating the Noether charge over this surface $\Sigma$, we get (for the calculation see p. 17-18 in \cite{Wald1994})
\begin{equation}
    \begin{split}
        \int_\Sigma \mathbf{Q}_N[V] &= \int_\Sigma \mathbf{X}^{cd} \nabla_{c} V_{d} + \ldots = \int_\Sigma \mathbf{X}^{cd} \mathcal{D}_{c} V_{d} \\ &= \int_\Sigma \mathbf{X}^{cd} \mathcal{D}_{[c} V_{d]} =  \int_\Sigma \mathbf{X}^{cd} \kappa \mathbf{b}_{cd} = \frac{\kappa}{2\pi} S, \\ &\text{ with } S = {2\pi} \int_\Sigma \mathbf{X}^{cd} \mathbf{b}_{cd},
    \end{split}
\end{equation} where $\ldots$ denote terms not contributing to the Noether charge at the bifurcation surface. Here we also used the fact that $2\mathcal{D}_{(c}V_{d)} = \mathcal{L}_V g_{cd} = 0$ for the Killing vector field $V$. We use the notation $\mathbf{b}_{cd}$ here for the binormal to $\Sigma$.
We thus get the following for the entropy and the first law.
\begin{thm} \label{th:palatinientropy}
The entropy $S$ of a black hole, defined by the integral of the Noether charge $\mathbf{Q}_N$ over the bifurcation surface $\Sigma$ \begin{equation}
        \int_\Sigma \mathbf{Q}_N[V] = \frac{\kappa}{2\pi} S,
\end{equation} can be expressed as \begin{equation} \label{eq:entropygeneralpalatini}
    \begin{split}
        S &= {2\pi} \int_\Sigma \mathbf{X}^{cd} \mathbf{b}_{cd}, \\ & \text{ where } (\mathbf{X}^{cd})_{a_3...a_n}=\varepsilon_{aba_3...a_n}{E_{R}}^{dcab}\ .
    \end{split}
\end{equation} For the definition of $E_R$ see Lemma \ref{lemma:palatiniboundary}.  
\end{thm}
\begin{thm} \label{th:firstlawBH}
Defining the energy, angular momentum and entropy as stated in the previous section and the above theorem, we get the first law of black hole thermodynamics \begin{equation}
     \frac{\kappa}{2\pi} \delta S = \delta E -   \Omega_H^{(\mu)} \delta \mathcal{J}_{(\mu)},
\end{equation} where the variation of the fields $\delta \phi$ satisfies the linearised equation of motion around an asymptotically flat stationary black hole solution having a bifurcate Killing horizon.
\end{thm}
\begin{remark}
In the boundary term $\bm{\Theta}$ there was an ambiguity  \begin{equation} 
    \mathbf{\Theta} \rightarrow \mathbf{\Theta} + \delta \bm{\mu},
\end{equation} leading to a contribution $\bm{Q}_N \rightarrow \bm{Q}_N + i_V \bm{\mu}$ and thus not contributing to the entropy, as argued in Theorem 6.1 in \cite{Wald1994} (namely due to the Killing vector field $V$ vanishing on the bifurcation surface). Thus the entropy is invariant under this ambiguity. 
\end{remark}
\begin{remark} \label{remark:summary_Yambiguity_Palatini}
As explained in Remark \ref{remark:Y_ambiguity} we have another ambiguity, namely the freedom the add an exact $(n-1)$-form to $\bm{\Theta}$: \begin{equation}
    \bm{\Theta} \rightarrow \bm{\Theta} + \mathrm{d}\mathbf{Y}(g,\Gamma,\delta g, \delta \Gamma),
\end{equation} with $\bm{Y}(g,\Gamma, \delta g, \delta \Gamma)$ linear in the variations $\delta g$ and $\delta \Gamma$, leading to an ambiguity in the Noether charge: \begin{equation}
    \mathbf{Q}_N \rightarrow \mathbf{Q}_N + \mathbf{Y}(g,\Gamma, \mathcal{L}_V g, \mathcal{L}_V \Gamma).
\end{equation} Terms in $\bm{Y}$ linear in $\mathcal{L}_V g$ will not contribute to the entropy $S$ or to $\delta S$, as $\mathcal{L}_V g = 0$ for the Killing vector field $V$, as argued on p. 17 of \cite{Wald1994}. However, terms in $\bm{Y}$ linear in $\mathcal{L}_V \Gamma$ might persist, and the above energy, angular momentum and entropy are defined up to this ambiguity. In STEGR however, as we shall see this ambiguity vanishes once the reduction of phase space to kinematically possible phase space is performed in Section \ref{sec:reduction_phasespace}.  The same conclusion holds for $\bm{Y}$-ambiguities concerning the right hand side of the first law, see also Remark \ref{remark:Y_ambig_symplecticform_Palatini}. 
\end{remark}

\begin{remark}
When one restricts the action to be `GR-type', that is, only depending on the Riemann tensor with zero torsion and non-metricity, one finds Wald's result for the entropy \cite{Wald1994}, as for the curvature of a metric-compatible torsion-free connection (the Levi-Civita connection) we have $R^{dcab} = R^{abdc}=-R^{abcd}$.
\end{remark}
\begin{remark}
In the above theorem the entropy is defined as the integral of the Noether charge over the bifurcation surface $\Sigma$ (up to some proportionality constants): \begin{equation}
    S = \frac{2\pi}{\kappa} \int_{\Sigma} \mathbf{Q}_N[V],
\end{equation} but for this actually any cross-section of a stationary horizon could have been chosen. For this we follow the argument as presented in \cite{Jacobson_1994}. Let $S_1$ and $S_2$ be two cross-sections. Considering integration over the piece of the horizon bounded by $S_1$ and $S_2$, we get that the difference between the integrals over $S_1$ and $S_2$ is equal to the integral of $\mathbf{J} = \mathrm{d}\mathbf{Q}_N$ over that piece of the horizon. But, for a stationary spacetime we have $\mathcal{L}_V g = 0$ and if the phase space reduction procedure yields $\mathcal{L}_V \Gamma = 0$ (which will be the case for STEGR) we have $\mathbf{J} = \bm{\Theta}(\mathcal{L}_V g, \mathcal{L}_V \Gamma)- i_V L = -i_V L$ and the pullback of $i_V L$ to the horizon vanishes as $V$ is tangent to the horizon. Thus we could as well define the entropy by \begin{equation}
    S = \frac{2\pi}{\kappa} \int_{\Xi} \mathbf{Q}_N[V],
\end{equation} for an arbitrary cross section $\Xi$ of the horizon.
\end{remark}
\begin{remark} \label{remark:notationintegralQJacobson}
Using a notation introduced in \cite{Jacobson_1994} we can write the expression of the entropy in \eqref{eq:entropygeneralpalatini} in another way which may be useful for some calculations. Here we will denote the binormal by $\mathbf{b}_{cd}$, whereas in the paper \cite{Jacobson_1994} the binormal is denoted by $\hat{\epsilon}_{cd}$. For the Levi-Civita tensor $\bm{\epsilon}_{a_1 \ldots a_n}$ we define the $(n-m)$-tensor $\bm{\epsilon}_{a_1 \ldots a_m}$ by \begin{equation}
    (\bm{\epsilon}_{a_1 \ldots a_m})_{a_{m+1} \ldots a_n} = \bm{\epsilon}_{a_1 \ldots a_n}.
\end{equation} Analogously we define the $(n-m)$-density $\varepsilon_{a_1 \ldots a_m}$ from the completely antisymmetric density $\varepsilon$. Using this notation, the entropy as in \eqref{eq:entropygeneralpalatini} becomes \begin{equation} \begin{split}
        S &= {2\pi} \int_\Sigma \mathbf{X}^{cd} \mathbf{b}_{cd} = {2\pi} \int_\Sigma E_R^{dcab} \varepsilon_{ab} \mathbf{b}_{cd} \\ &= {2\pi} \int_\Sigma \big{(}\frac{E_R^{dcab}}{\sqrt{-g}}\big{)} \mathbf{b}_{cd} \bm{\epsilon}_{ab} = {2\pi} \int_\Sigma \big{(}\frac{E_R^{dcab}}{\sqrt{-g}}\big{)} \mathbf{b}_{cd} \mathbf{b}_{ab} \overline{\epsilon},
        \end{split}
\end{equation} with $\overline{\epsilon}$ is the induced volume form on the bifurcation surface $\Sigma$ and where we used the equality $\bm{\epsilon}_{a b}=\mathbf{b}_{a b} \bar{\epsilon}$, which holds when pulled back to the bifurcation surface we integrate over \cite{Jacobson_1994}. Recall that $E_R^{abcd}$ is a density of weight $-1$, so that $\frac{1}{\sqrt{-g}}E_R^{abcd}$ is a tensor.
\end{remark}
\begin{remark}
If we consider theories with vanishing curvature ${R^a}_{bcd}=0$, the entropy will have a simple form. We split the Lagrangian into a curvature independent part $\mathcal{L}^{(0)}$ - which then depends on the remaining quantities in \eqref{eq:Palatiniaction1} - plus a constraint setting the curvature tensor to zero:
\begin{equation}
    \mathcal{L}=\mathcal{L}^{(0)}+{\lambda_a}^{bcd}{R^a}_{bcd}\ .
\end{equation}
As we have just shown, the only contribution to the entropy comes from varying $\mathcal{L}$ with respect to the curvature tensor $R_{abcd}$, which thus in this case leads to
\begin{equation}
(\mathbf{X}^{cd})_{a_3...a_n}=\varepsilon_{aba_3...a_n} \lambda^{dcab} \ .
\end{equation}
The entropy is thus given \textit{entirely} by the Lagrange multiplier! Its value is determined by the connection equation of motion; if we call ${E^{(0)}_{\Gamma a}}^{bc}$ the connection equation of motion coming from $\mathcal{L}^{(0)}$ (including hypermomentum), then we have the complete connection equation of motion
\begin{equation}
\label{eq:PalatiniLM}
    2\nabla_d {\lambda_a}^{cdb}={E^{(0)}_{\Gamma a}}^{bc}\ .
\end{equation}
Since ${\lambda_a}^{cdb}$ inherits the antisymmetry of the Riemann tensor in the last two indices, differentiating \eqref{eq:PalatiniLM} w.r.t. $\nabla_b$ will make the left side vanish, leading to $\nabla_b {E^{(0)}_{\Gamma a}}^{bc}=0$. Note that $\lambda_R$ might not be determined uniquely by the connection equation of motion \eqref{eq:PalatiniLM}; there might be additional symmetries involving $\lambda^{abcd}$ in the action. These should not change the entropy to have a well defined entropy; the entropy should be uniquely defined and independent of the additional symmetries of the Lagrange multipliers. This is investigated for example in the case of STEGR below.
\end{remark}

\vspace{0.5cm}

\section{Calculating entropy in STEGR using the extended Wald Noether charge method} \label{sec:STEGR}

\subsection{Entropy in Symmetric Teleparallelism Gravity a la Palatini}
As a first non-trivial example, let us take a look at non-metric theories, i.e. we set the curvature $R(\Gamma)$ and torsion $T$ to zero, and let the Lagrangian be an arbitrary function of $Q_{abc}$ (and possibly its derivatives). This is also called STG or Symmetric Teleparallel Gravity. The action in the Palatini formalism has the form
\begin{equation}
    S_{STG}=\int d^4x \left[ \sqrt{-g} f(g_{ab},Q_{abc})+\lambda_{Ra}{}^{bcd} R^{a}{}_{bcd} + \lambda_{Ta}{}^{bc}T^{a}{}_{bc} \right] +S_M\ ,
\end{equation}
where $S_M$ is the action of extra matter fields, $f$ is an arbitrary function of weight 0, and $\lambda_{Ra}{}^{bcd}$ and $\lambda_{Ta}^{bc}$ are Lagrange multipliers, tensors of weight $-1$ with the same symmetries as the Riemann and torsion tensors, respectively, namely antisymmetry in their last two indices. These Lagrange multipliers are restricted by the connection equation of motion \begin{equation}
    \nabla_d{\lambda_{Ra}}^{cbd}+{\lambda_{Ta}}^{bc}=\sqrt{-g}{P^{bc}{}_a}+{\mathcal{H}_a}^{bc}\ ,
\end{equation} with the non-metricity conjugate defined by $\sqrt{-g} P^{a}{}_{bc}=\frac{\partial\mathcal{L}}{\partial Q_{a}{}^{bc}}$ and the hypermomentum current defined by \begin{equation}
\mathfrak{H}_{c}{}^{a b}=-\frac{1}{2} \frac{\delta \mathcal{S}_{M}}{\delta \Gamma_{a b}^{c}}.
\end{equation} As pointed out in \cite{Heisenberg18}, the gauge symmetries of the Lagrange multipliers are given by \begin{equation}
\begin{aligned}
\lambda_{Ra}{}^{bcd} & \rightarrow \lambda_{Ra}{}^{bcd}+\left(\nabla_{e}+T_{e}\right) \kappa_a{}^{bcde}+T^c{}_{ef} \kappa_a{}^{bfde}+\kappa_a{}^{bcd} \\
\lambda_{Ta}{}^{bc} & \rightarrow \lambda_{Ta}{}^{bc}+\left(\nabla_{d}+T_{d}\right) \kappa_a{}^{bcd}+T^{b}{}_{ed} \kappa_a{}^{ecd},
\end{aligned}
\end{equation}  where $\kappa_{a}{}^{bcde}$ and $\kappa_{a}{}^{bcd}$ are tensor densities, totally antisymmetric in their last three indices. To have a uniquely defined entropy we must make sure the entropy does not change under these transformations. We can restrict to the transformations when the equations of motion are obeyed (on-shell), with thus also the constraints of zero curvature and zero torsion obeyed, as the entropy is evaluated on the fields obeying the equations of motion. 

\begin{thm} \label{th:entropySTG}
In STG the boundary term contributing to the entropy is equal to \begin{equation} \label{eq:boundarySTG}
    \Theta^a = 2{\lambda_{Rd}}^{c a b}\delta \Gamma^d_{bc},
\end{equation} leading to an entropy $S$ \begin{equation}
    \begin{split}
        S &= \frac{1}{2\pi} \int_\Sigma \mathbf{X}^{cd} \mathbf{b}_{cd}, \text{ where } (\mathbf{X}^{cd})_{a_3...a_n}=\varepsilon_{aba_3...a_n}\lambda_R^{dcab}\ .
    \end{split}
\end{equation}  This entropy is uniquely defined, that is, it is invariant under the on-shell symmetry transformations \begin{equation} \label{eq:symmlambdaR}
\begin{split}
\lambda_{Ra}{}^{bcd} & \rightarrow \lambda_{Ra}{}^{bcd}+\nabla_{e} \kappa_a{}^{bcde}+\kappa_a{}^{bcd},
\end{split}
\end{equation} where $\kappa_{a}{}^{bcde}$ and $\kappa_{a}{}^{bcd}$ are tensor densities, totally antisymmetric in their last three indices.
\end{thm}
\begin{proof}
The first part of the proof, the form of the boundary term and of the entropy, directly follow from Lemma \ref{lemma:palatiniboundary} and Theorem \ref{th:palatinientropy}. What is left to prove is the invariance of the entropy under the transformations in \eqref{eq:symmlambdaR}. 

\vspace{0.5cm}

Under the symmetry \begin{equation}
    \lambda_{Ra}{}^{bcd} \rightarrow \lambda_{Ra}{}^{bcd} + \kappa_{a}{}^{bcd},
\end{equation} with $k_{a}{}^{bcd}$ totally antisymmetric in the last three indices, the entropy will not change, as we would get an extra term in the boundary of the form \begin{equation} 2{\kappa_{Rc}}^{b a d}\delta\Gamma^c_{db} = 0, \end{equation} where the equality is due to the fact that the variation of the connection is torsion-free, that is $\delta \Gamma^c_{db}$ is symmetric in $d$ and $b$ while $\kappa_c^{bad}$ is antisymmetric in $b$ and $d$.

\vspace{0.5cm}

Under the symmetry \begin{equation} \label{eq:addnablakappaSTG}
    \lambda_{Ra}{}^{bcd} \rightarrow \lambda_{Ra}{}^{bcd} + \nabla_e \kappa_{a}{}^{bcde},
\end{equation}  the Lagrangian changes by a term \begin{equation}
    \nabla_e \kappa_{a}{}^{bcde} R^a{}_{bcd} = \nabla_e( \kappa_{a}{}^{bcde} R^a{}_{bcd}) - \kappa_{a}{}^{bcde}\nabla_e R^a{}_{bcd}.
\end{equation} The first term just adds a total derivative to the Lagrangian which would lead to a term $\delta \bm{\mu}$ in the boundary under a variation; this will not change the entropy, see for example Proposition 4.1 in \cite{Wald1994}. Focusing on the second term, varying with respect to $\Gamma$ gives terms like \begin{equation}
    \kappa_{a}{}^{bcde}\delta \Gamma^a_{ef} R^f{}_{bcd} + \ldots,
\end{equation} contributing only to the connection equation of motion, and the term \begin{equation}
    \kappa_{a}{}^{bcde}\nabla_e \delta R^a{}_{bcd} = \kappa_{a}{}^{bcde} \nabla_e \nabla_{[c} \delta \Gamma^a_{d]b}.
\end{equation} 

Let us calculate the term appearing at the right hand side of the above equation: 
\begin{equation}
    \begin{split}
        & \nabla_e \nabla_{c} \delta \Gamma^a_{db} = \nabla_e (\partial_c \delta \Gamma^a_{db} + \Gamma^a_{c f} \delta \Gamma^f_{db} - \Gamma^f_{cd} \delta \Gamma^a_{fb}-\Gamma^f_{cb} \delta \Gamma^a_{df}) - ( e \leftrightarrow c) \\ &= \big{(} \partial_e \partial_c \delta \Gamma^a_{db}-\Gamma^f_{ec} \partial_f \delta \Gamma^a_{db}-\Gamma^f_{ed}\partial_c \delta \Gamma^a_{fb}- \Gamma^f_{eb} \partial_c \delta \Gamma^a_{df}+\Gamma^a_{ef} \partial_c \delta \Gamma^f_{db} \big{)} \\ & + \big{(} \partial_e \Gamma^a_{cf} \delta \Gamma^f_{db} + \Gamma^a_{cf} \partial_e \delta \Gamma^f_{db} + \Gamma^a_{eg} \Gamma^g_{cf} \delta \Gamma^f_{db}  - \Gamma^g_{ec} \Gamma^a_{gf} \delta \Gamma^f_{db} - \Gamma^g_{ed} \Gamma^a_{cf} \delta \Gamma^f_{gb} - \Gamma^g_{eb} \Gamma^a_{ef} \delta \Gamma^f_{dg} \big{)}  \\ &+  \big{(} -\partial_e \Gamma^f_{cd} \delta \Gamma^a_{fb} - \Gamma^f_{cd} \partial_e \delta \Gamma^a_{fb} - \Gamma^a_{eg} \Gamma^f_{cd} \delta \Gamma^g_{fb}  + \Gamma^g_{ec} \Gamma^f_{gd} \delta \Gamma^a_{fb} + \Gamma^g_{ed} \Gamma^f_{cg} \delta \Gamma^a_{fb} + \Gamma^g_{eb} \Gamma^f_{cd} \delta \Gamma^a_{fg} \big{)} \\ &+ \big{(} - \partial_e \Gamma^f_{cb} \delta \Gamma^a_{df} - \Gamma^f_{cb} \partial_e \delta \Gamma^a_{df} - \Gamma^a_{eg} \Gamma^f_{cb} \delta \Gamma^g_{df}  + \Gamma^g_{ec} \Gamma^f_{gb} \delta \Gamma^a_{df} + \Gamma^g_{ed} \Gamma^f_{cb} \delta \Gamma^a_{gf} + \Gamma^g_{eb} \Gamma^f_{cg} \delta \Gamma^a_{df} \big{)} \\ &- ( e \leftrightarrow c)      \\ &= -T^g{}_{ec} \nabla_g \delta \Gamma^a_{db} + R^a{}_{fec} \delta \Gamma^f_{db} - R^f{}_{dec} \delta \Gamma^a_{fb} - R^f{}_{bec} \delta \Gamma^a_{df},
    \end{split}
\end{equation} and thus under the transformation in \eqref{eq:addnablakappaSTG} we obtain an extra boundary term \begin{equation}
    \Theta^f = \kappa_a{}^{bcde} T^f{}_{e[c} \delta \Gamma^a_{d]b},
\end{equation} which will be equal to zero on the equations of motion as the torsion $T^f{}_{ec}$ is zero there, and thus not contributing to the entropy. We conclude that the entropy is unchanged under the transformation in \eqref{eq:addnablakappaSTG}.
\end{proof}

\begin{remark}
Since the boundary term is
\begin{equation}
     \Theta^a = 2{\lambda_{Rc}}^{b a d}\delta \Gamma^c_{db},
\end{equation}
only the part of $\lambda_R$ symmetric in $b$ and $d$ will enter. Symmetrizing the connection equation of motion removes $\lambda_T$,
\begin{equation}
\label{eq:STboundary1}
    \nabla_d{\lambda_{Ra}}^{(cb)d}=\sqrt{-g}{P^{(bc)}}_a ,
\end{equation}
where we have also set hypermomentum to zero, and used the antisymmetry of $\lambda_R$. 
\end{remark}

\begin{remark}
We could also prove the above invariance of the entropy under the symmetries in $\lambda_R$ by directly looking at what changes in the boundary term from the newly obtained action and how the on-shell current $\mathbf{J}$ changes. If the on-shell current is unchanged, by definition the Noether charge $\mathbf{Q}_N$ will be unchanged as well. Thus if symmetries in $\lambda_R$ keep $\mathbf{J}$ unchanged, the entropy will be unchanged. 
\end{remark}

\begin{remark}
We should note that up to this point we still have the terms linear in $\delta \Gamma$ in the $\bm{Y}$-ambiguity which might contribute to the entropy, see also Remark \ref{remark:summary_Yambiguity_Palatini}. However, in STEGR for example this ambiguity will vanish due to a reduction in phase space. 
\end{remark}

\vspace{0.5cm}

\subsection{Entropy in STEGR a la Palatini}

Now let us return to the special case of the STEGR action (conventions of \cite{Heisenberg18}), \begin{equation} \label{eq:STEGRpalatini2}
\mathcal{S}_{\mathbb{Q}}=\int \mathrm{d}^{4} x\left[- \frac{1}{2}\sqrt{-g} \mathbb{Q}+\lambda_{Ra}{}^{bcd} R^{a}{}_{bcd}+\lambda_{Ta}{}^{bc} T^{a}{}_{bc}\right], \end{equation} where the quadratic form $\mathbb{Q}$ of the non-metricity is given by \begin{equation}
\mathbb{Q}=\frac{1}{4} Q_{a b c} Q^{a b c}-\frac{1}{2} Q_{a b c} Q^{b a c}-\frac{1}{4} Q_{a} Q^{a}+\frac{1}{2} Q_{a} \tilde{Q}^{a}.
\end{equation} Recall that for this quadratic form we had the relation: \begin{equation}
R=\mathcal{R}(g)+\mathbb{Q}+\mathcal{D}_{a}\left(Q^{a}-\widetilde{Q}^{a}\right),
\end{equation} so that we have the same equations of motion as in GR from the Einstein-Hilbert action, and thus we would expect the same expression for the entropy. As described in the previous section, for STEGR we will have an entropy \begin{equation}   \label{eq:entropySTEGR} \begin{split} S &= \frac{1}{2\pi} \int_\Sigma \mathbf{X}^{cd} \mathbf{b}_{cd},  \\ & \text{with } (\mathbf{X}^{cd})_{fg} = \lambda_R^{dcab} \bm{\epsilon}_{abfg},   \end{split}
\end{equation}
where $\lambda_R$ is the Lagrange multiplier requiring $R_a{}^{bcd} = 0$ as described in previous sections, and has to be evaluated on the equations of motion. Note that $\lambda_R$ and $\lambda_T$ are tensor densities of weight -1 \cite{HeisenbergL18}, so that $\frac{\lambda_R}{\sqrt{-g}} = \Lambda_R$ and $\frac{\lambda_T}{\sqrt{-g}} = \Lambda_T$ are tensors. Recall also that $ \lambda_{Ra}{}^{cbd}$ and $ \lambda_{Ta}{}^{bd}$ are antisymmetric in their last two indices $b$ and $d$. The connection equation of motion is given by (from \cite{Heisenberg18}) \begin{equation} \label{eq:lambdaRlambdaTP}
\nabla_{d} \lambda_{Ra}{}^{cbd}+\lambda_{Ta}{}^{b c}= \sqrt{-g} P^{b c}{}_a,
\end{equation} where $P^{bc}_a$ is defined by $\sqrt{-g}P^a{}_{bc} = \frac{\partial \mathcal{L}}{\partial Q_a{}^{bc}}$, which takes the following form here:  \begin{equation}
    \label{eq:form_nonmetrconj_P}  P^{bc}{}_a = -\frac{1}{4} Q^{bc}{}_a +\frac{1}{4}Q^{cb}{}_a{} + \frac{1}{4}Q_a{}^{bc} + \frac{1}{4}\delta^c_a(Q^b-\widetilde{Q}^b) - \frac{1}{8}g^{bc}Q_a - \frac{1}{8}\delta^b_a Q^c.
\end{equation} For this connection equation of motion the torsion-free constraint has already been used \cite{Heisenberg18}. Taking $\nabla_b$ of the connection equation of motion and using $\nabla_b \nabla_d  \lambda_{Ra}{}^{cbd} = \nabla_{(b} \nabla_{d)}  \lambda_{Ra}{}^{cbd} = 0$, where the first equality holds as the curvature of $\nabla$ is zero and we obtain the last equality from the antisymmetry of $ \lambda_{Ra}{}^{cbd}$ in $b$ and $d$, we get the equation \begin{equation} \label{eq:lambdaTP}
\nabla_{b} \lambda_{Ta}{}^{bc}=\nabla_b ( \sqrt{-g} P^{b c}{}_a).
\end{equation} Thus, to find $\lambda_R$ and $\lambda_T$, we have to solve \eqref{eq:lambdaRlambdaTP}, keeping in mind that $\lambda_R$ and $\lambda_T$ have to be antisymmetric in their last two indices. To show that the entropy we will obtain here will be the same as in GR from the Einstein-Hilbert action, we will find one solution of the above equations, and as we have argued in Theorem \ref{th:entropySTG}, the entropy will not change under symmetry transformations of the Lagrange multipliers, showing the uniqueness of the entropy. 

\vspace{0.5cm}

\begin{thm}
The connection equation of motion \eqref{eq:lambdaRlambdaTP} has a solution $\lambda_R,\lambda_T$ where $\Lambda_R = \frac{\lambda_R}{\sqrt{-g}}$ is given by \begin{equation}
    \label{eq:LambdaRentropyform} 
    \Lambda_{Ra}{}^{cbd} = \frac{1}{4}(-g^{bc}\delta^d_a + g^{cd}\delta^b_a). 
\end{equation} Note that this $\Lambda_R$ is indeed antisymmetric in the last two indices $b$ and $d$. From the entropy expression above (see \eqref{eq:entropySTEGR}) or the form of the boundary terms (see \eqref{eq:boundarySTG}) we see that indeed this leads to the same entropy as in GR from the Einstein-Hilbert action. 
\end{thm}

\begin{proof} The idea of the proof goes as follows. We take $\Lambda_R$ as in \eqref{eq:LambdaRentropyform}, and then we define $\lambda_T$ from the connection equation of motion \eqref{eq:lambdaRlambdaTP}. It will turn out that indeed this $\lambda_T$ will be antisymmetric in its last two arguments, as required. Thus from Theorem \ref{th:entropySTG} we would have a boundary term (up to terms not contributing to the entropy) equal to, \begin{equation} \begin{split}
    \Theta^a = 2 \lambda_{Rc}{}^{bad} \delta \Gamma^c_{db} &= \frac{1}{2}(-g^{ab}\delta^d_c + g^{bd} \delta^a_c ) \delta \Gamma_{db}^c  \\ &= \frac{1}{2}(-g^{ab} \delta \Gamma^d_{db} + g^{bd} \delta \Gamma^a_{db}),
\end{split}
\end{equation} leading thus to the same entropy as in GR from the EH action; see for example Appendix \ref{appendix:boundarytermSTEGR_same_entropy_GR} for this. For an expression of the exact boundary term (up to the $\bm{Y}$-ambiguity, see Remark \ref{remark:boundary_terms_STEGR}. 

\vspace{0.5cm}

Now let us thus calculate the term $\nabla_{d} \lambda_{Ra}{}^{cbd}$: \begin{equation} \label{eq:nablalambdaR}
    \begin{split}
       \frac{1}{\sqrt{-g}} \nabla_{d} \lambda_{Ra}{}^{cbd} &= \frac{1}{\sqrt{-g}}  \nabla_{d} (\sqrt{-g}\Lambda_{Ra}{}^{cbd}) = (\nabla_d + \frac{Q_d}{2}) \Lambda_{Ra}{}^{cbd} \\ &= (\mathcal{D}_d + \frac{Q_d}{2})\Lambda_{Ra}{}^{cbd} - L^e{}_{da} \Lambda_{Re}{}^{cbd} + L^c{}_{de} \Lambda_{Ra}{}^{ebd} + L^b{}_{de} \Lambda_{Ra}{}^{ced} + L^d{}_{de} \Lambda_{Ra}{}^{cbe} \\ &=  \mathcal{D}_d \Lambda_{Ra}{}^{cbd} - L^e{}_{da} \Lambda_{Re}{}^{cbd} + L^c{}_{de} \Lambda_{Ra}{}^{ebd} + L^b{}_{de} \Lambda_{Ra}{}^{ced},
    \end{split}
\end{equation} where in the first and third line we used the identities $\nabla_d \sqrt{-g} = \sqrt{-g}\frac{Q_d}{2}$ and $L^d{}_{de} = -\frac{Q_e}{2}$ respectively. Taking now for $\Lambda_R$ the expression from \eqref{eq:LambdaRentropyform}, we get \begin{equation} \label{eq:nablalambdaR2}
    \begin{split}
      4 \frac{1}{\sqrt{-g}} \nabla_{d} \lambda_{Ra}{}^{cbd} &= - L^e{}_{da} (-g^{bc}\delta^d_e + g^{cd}\delta^b_e) + L^c{}_{de} (-g^{be}\delta^d_a + g^{ed}\delta^b_a) + L^b{}_{de}(-g^{ec}\delta^d_a + g^{cd}\delta^e_a) \\ &= L^d{}_{da}g^{bc}-L^{bc}{}_a - L^{cb}{}_a + L^{cd}{}_d \delta^b_a - L^{bc}{}_a + L^{bc}{}_a \\ &= -\frac{Q_a}{2}g^{bc}-L^{bc}{}_a - L^{cb}{}_a + (\frac{Q^c}{2}-\widetilde{Q}^c) \delta^b_a  \\ &= Q_{a}{}^{bc}-\frac{1}{2}Q_a g^{bc}+(\frac{1}{2}Q^c-\widetilde{Q}^c) \delta^b_a. 
    \end{split}
\end{equation} For $P^{bc}{}_a$ we have 
\begin{equation} \label{eq:Pbca}
    \begin{split}
     P^{bc}{}_a &=  \frac{1}{\sqrt{-g}}  \frac{\partial \mathcal{L}_{STEGR}}{\partial Q_{bc}{}^a} \\ &= -\frac{1}{4} Q^{bc}{}_a +\frac{1}{4}Q^{cb}{}_a{} + \frac{1}{4}Q_a{}^{bc} + \frac{1}{4}\delta^c_a(Q^b-\widetilde{Q}^b) - \frac{1}{8}g^{bc}Q_a - \frac{1}{8}\delta^b_a Q^c.
    \end{split}
 \end{equation}
 
We thus would get the following expression for $\Lambda_T = \frac{\lambda_T}{\sqrt{-g}}$ from the connection equation of motion \eqref{eq:lambdaRlambdaTP}: 
\begin{equation} \label{eq:lambdaTexprSTEGR}
    \begin{split}
        \Lambda_{Ta}{}^{bc} &= P^{bc}{}_a -  (\nabla_d + \frac{Q_d}{2}) \Lambda_{Ra}{}^{cbd} \\ &= \frac{1}{2} Q^{[cb]}{}_a - \frac{1}{2} \delta^{[c}_a \widetilde{Q}^{b]} + \frac{1}{2}\delta^{[c}_a Q^{b]}.
    \end{split} 
\end{equation} Note that $\Lambda_{Ta}{}^{bc}$ is indeed antisymmetric in the last two indices $b$ and $c$, as required. This ends our proof.

\end{proof}

\begin{remark}
In general, to obtain an expression for $\Lambda_R$, one could try and solve the symmetrized connection equation of motion (see \eqref{eq:STboundary1}) which then only involves $\lambda_R$ by guessing the form of $\Lambda_R$ by the form of $P_{bc}{}_a$. For theories quadratic in $Q_{abc}$ one could expect $\Lambda_T$ to be linear in $Q$ and its contractions, and we might expect $\Lambda_R$ to only contain combinations of the metric and kronecker delta functions.
\end{remark}

\begin{remark}
We should note that up to this point we still have the terms linear in $\delta \Gamma$ in the $\bm{Y}$-ambiguity which might contribute to the entropy, see also Remark \ref{remark:summary_Yambiguity_Palatini}. As is shown in the next section, in STEGR this ambiguity will vanish due to a reduction in phase space; the physical phase space will consist of the metric and its conjugate momentum, while the connection directions $\delta \Gamma$ will turn out to be degenerate directions of the presymplectic form and thus quotiented out in the physical phase space.
\end{remark}

\begin{remark} \label{remark:boundary_terms_STEGR}
For future reference, let us explicitly write out the boundary terms in STEGR. From the Palatini action, together with \eqref{eq:variation_of_R_Q} we have that the terms contributing to the boundary term are \begin{equation} 
    \begin{split}
        \nabla_{d} (\delta \Gamma^{a}{}_{c b}2\lambda_{Ra}{}^{b d c} -\delta g^{b c} \sqrt{-g} P^{d}{}_{b c}),
    \end{split}
\end{equation} which together with the expression for $\lambda_R$ found in the above theorem we get \begin{equation} \label{eq:exact_BT_STEGR}
    \Theta^d = g^{b[c} \delta \Gamma^{d]}{}_{c b} + P^{dbc} \delta g_{b c}.
\end{equation}
\end{remark}

\begin{remark}
The entropy can also be directly obtained from the calculation of the Noether charge, which is performed in Appendix \ref{appendix:Noethercharge_STEGR}. From this expression it is also clear that the entropy in STEGR will be the same as in GR. 
\end{remark}

\vspace{0.5cm}

\subsection{Reduction of phase space in STEGR} \label{sec:reduction_phasespace}

Let us calculate the phase space of STEGR. We start by calculating the boundary term $\theta^\mu$, arising from the variation of the Lagrangian density $\delta \mathcal{L} = E \delta \phi + \nabla_\mu \theta^{\mu}(\phi,\delta \phi)$. We will heavily rely on ideas of \cite{Wald1990}. What we do next is to assume a foliation by Cauchy surfaces $\Sigma_t$, parametrised by some `time' coordinate $t$, with unit normal $n^a$. We consider the lapse function $N$ and shift $N^i$ as in the ADM decomposition of GR, so that we can describe our metric $g_{ab}$ with $N, N^i$ and an induced spatial metric $h_{ij}$ on $\Sigma_t$. This will turn out to be sufficient to have a representation of our physical phase space. We consider variations on $\Sigma_t$ to obtain our phase space, so that our unit normal $n^a$ remains unchanged (as it is defined in the complement of $\Sigma_t$). Furthermore, we will assume $\Gamma$ to be symmetric, so that our connection is already torsion-free. In this setting we have the following. \begin{lemma} \label{lemma:Gamma_degeneracy_presymplectic} In STEGR, the variations $\delta \Gamma$ are degeneracy directions of the presymplectic form $\bm{\omega}$.
\end{lemma}
\begin{proof} For a variation of $\mathcal{L}_{STEGR} = -\frac{1}{2}\sqrt{-g}\mathbb{Q}+\lambda_{Ra}{ }^{b cd} R^{a}{}_{b cd}+\lambda_{Ta}{ }^{bc} T^{a}{}_{bc}$ with respect to the connection $\Gamma$ and spatial metric $h$ we get: \begin{equation}
    \begin{split}
        \delta_{\Gamma} \mathcal{L} = E \delta \Gamma + \nabla_{d} (\delta \Gamma^{a}{}_{c b}2\lambda_{Ra}{}^{b d c}), \\
        \delta_g \mathcal{L} = E \delta g + \nabla_{d}(-\delta g^{b c} \sqrt{-g} P^{d}{}_{b c}),
    \end{split} 
\end{equation} where $E$ denote the equations of motion (suppressing indices in the first term), and we used the formulas for the variations of $R$ and $Q$ as in \eqref{eq:variation_of_R_Q}. This leads to a presymplectic current $\omega^d = \delta_1 \theta^{d}_2 - \delta_2 \theta^{d}_1$ (here $\theta_i = \theta(\phi,\delta_i \phi)$) of the following form \begin{equation} \begin{split}
    \omega^{d} n_d = \delta_2 (\Gamma^{a}{}_{c b}) \delta_1 (2\lambda_{Ra}{}^{b d c} n_d) - \delta_1 (\Gamma^{a}{}_{c b}) \delta_2 (2\lambda_{Ra}{}^{b d c} n_d) \\ - \delta_2 (g^{b c}) \delta_1 (\sqrt{-g}P^{d}{}_{b c} n_d)+ \delta_1 (g^{b c}) \delta_2 (\sqrt{-g} P^{d}{}_{b c} n_d), \end{split}
\end{equation} where for a variation $\delta g$ only variations of the spatial metric $h$ are meant, so the unit normal is left unchanged. This leads then to a presymplectic form (which is defined by $\omega = \int_{\Sigma} \omega_d n^d$) given by \begin{equation}\begin{split}
     \omega\bigg((\delta_1 h, \delta_1 \Gamma),(\delta_2 h, \delta_2 \Gamma) \bigg) = \int_{\Sigma} \bigg( \delta_2 (\Gamma^{a}{}_{c b}) \delta_1 (\pi_{\Gamma}{}_{a}{}^{b c}) - \delta_1 (\Gamma^{a}{}_{c b}) \delta_2 (\pi_{\Gamma}{}_a{}^{b c}) \\ + \delta_2 (g^{b c}) \delta_1 (\pi_{g}{}_{b c})- \delta_1 (g^{b c}) \delta_2 (\pi_{g}{}_{b c}) \bigg)
\end{split}
\end{equation} with \begin{equation}
\begin{split}
    \pi_{\Gamma}{}_{a}{}^{b c} = 2\lambda_{Ra}{}^{b d c} n_d, \\
    \pi_{g}{}_{b c} = -\sqrt{-g} P^{d}{}_{b c} n_d.
\end{split}    
\end{equation} 

As we are interested in a theory which satisfies constraints, we will look how the phase space gets reduced from these constraints. Let us look at STEGR where the constraints are imposed, that is from now on we assume the connection $\Gamma$ to be curvature- and torsion-free. As the Lagrange multiplier $\lambda_R$ is unphysical, we let it already satisfy its equations of motion, that is $\lambda_{Ra}{}^{cbd} =\frac{1}{4} \sqrt{-g} (-g^{bc}\delta^d_a + g^{cd}\delta^b_a)$ (see \eqref{eq:LambdaRentropyform} above). Note that the boundary term is unchanged under the symmetry transformation of $\lambda_R$, as proved in Theorem \ref{th:entropySTG}. This means that \begin{equation}
    \pi_{\Gamma}{}_a{}^{cd} =  \frac{1}{2}\sqrt{-g}(-n^c \delta^d_a + g^{cd} n_a),
\end{equation} so that \begin{equation} \begin{split}
    \delta \pi_{\Gamma}{}_a{}^{cd} &= \frac{1}{2}\sqrt{-g}n_a \delta g^{cd}-\frac{1}{4}\sqrt{-g} g_{ef}(-n^c \delta^d_a + g^{cd}n_a) \delta g^{ef}.
    \end{split}
\end{equation} With this substitution the presymplectic form becomes \begin{equation}
    \begin{split}
        \omega = \int_{\Sigma} &\bigg[ \delta_2  g^{bc}\bigg(\delta_1\pi_{g}{}_{bc}-\frac{1}{2}\sqrt{-g}\delta_1 \Gamma^{a}{}_{cb}n_a +\frac{1}{4}\sqrt{-g}g_{bc} \delta_1 \Gamma^a{}_{ed}(-n^d \delta^e_a + g^{de}n_a)\bigg) \\ &- \delta_1  g^{bc}\bigg(\delta_2\pi_{g}{}_{bc}-\frac{1}{2}\sqrt{-g}\delta_2\Gamma^{a}{}_{cb}n_a +\frac{1}{4}\sqrt{-g}g_{bc} \delta_2 \Gamma^a{}_{ed}(-n^d \delta^e_a + g^{de}n_a) \bigg) \bigg].
    \end{split}
\end{equation} Now for this presymplectic form, field variations $(\delta g, \delta \pi_g, \delta \Gamma)$ form a degeneracy direction if we take \begin{equation}
    \begin{split}
        \delta_2 g^{bc} &= 0, \\
        \delta_2 \pi_{g}{}_{bc} &= \frac{1}{2}\sqrt{-g}\delta_2 \Gamma^a_{cb} n_a-\frac{1}{4}\sqrt{-g}g_{bc} \delta_2 \Gamma^a{}_{ed}(-n^d \delta^e_a + g^{de}n_a).
    \end{split}
\end{equation} (And we let the $\delta_1$ variations be arbitrary.) Thus we can obtain degeneracy directions by letting $\Gamma$ vary arbitrarily, and choosing variations of $g, \pi_g$ as prescribed above. Hence we can identify phase space with the manifold of pairs $(g,\pi_g)$, with symplectic form \begin{equation}
    \Omega\bigg((\delta_1 g,\delta_1 \pi_g),(\delta_2 g, \delta_2 \pi_g) \bigg) = \int_{\Sigma} \bigg(\delta_2 (g^{b c}) \delta_1 (\pi_{g}{}_{b c})- \delta_1 (g^{b c}) \delta_2 (\pi_{g}{}_{b c}) \bigg).
\end{equation}
\end{proof}

\begin{remark}
Another way to get to a reduced phase space is the following. In STEGR with constrained connection we have \begin{equation}
    \mathbf{L}_{STEGR} = \mathbf{L}_{EH} + \mathrm{d}\bm{\mu}.
\end{equation} Taking variations, we get \begin{equation} \begin{split}
    \delta \mathbf{L}_{STEGR} = \delta\mathbf{L}_{EH} + \mathrm{d}\delta \bm{\mu} = \mathbf{E}_{EH}+\mathrm{d}( \bm{\Theta}_{EH}+\delta \bm{\mu} ), 
\end{split} 
\end{equation} where $\mathbf{E}_{EH}$ are the Einstein-Hilbert equations of motion. The symplectic form is then given by (in coordinates, see equation (2.19) in \cite{Wald1990}) \begin{equation} 
    \omega^a = \delta_1 \theta_2^a - \delta_2 \theta_1^a,
\end{equation} for two variations of the fields $\phi(\lambda_1,\lambda_2)$, where the boundary term $\theta_i^a$ is the dual density of the boundary term $\bm{\Theta}(\phi,\delta_i \phi)$, which has been obtained from the variation $\delta_i \bm{L}$ of the lagrangian. From this formula we see that the $\delta \bm{\mu}$-term won't contribute to $\omega$: \begin{equation}
    \omega^{a} = \delta_1 \delta_2 \mu^a - \delta_2 \delta_1 \mu^a = 0, 
\end{equation} because of commutativity of partial derivatives ($\delta_i \mu = \frac{\partial \mu}{\partial \lambda_i})$. Thus we conclude that \begin{equation}
    \bm{\omega}_{STEGR} = \bm{\omega}_{EH},
\end{equation} which we will just name $\bm{\omega}$ now. The presymplectic form $\bm{\omega}_{EH}$ obtained from the Einstein-Hilbert action will not involve the connection $\Gamma$, and thus every `connection direction' in the configuration space (or in the phase space) will be a degenerate direction for the presymplectic form $\bm{\omega}$. This means that in order to get the `kinematically possible' phase space (using the language of \cite{Wald1990}) we have to quotient out all connection directions in phase space.
\end{remark}

\vspace{0.5cm}

\subsection{First law of black hole thermodynamics in STEGR}
The first law derivation demands that we have
\begin{equation} \label{eq:omega=0}
    \omega(\phi,\delta\phi,\Lie\phi)=0\ ,
\end{equation}
where $\omega$ is defined as in Definition \ref{def:omega} and $\phi$ denotes all the dynamical fields, including the connection. We could require $\Lie\phi=0$ for all $\phi$, but we want to see what happens if we demand $\Lie g_{ab}=0$ only. In GR this immediately implies $\Lie\Gamma_{LC}=0$, but for general connections this does not hold. However, as shown in the previous section, we will have $\mathcal{L}_V \Gamma = 0$ in reduced phase space, which thus then yields $\omega(\phi,\delta \phi, \mathcal{L}_V \phi) = 0$ in reduced phase space. Alternatively, from the STEGR boundary term \begin{equation}
    \Theta^a=\sqrt{-g}\left(-g^{c[a}\delta^{b]}_d\delta{\Gamma^d}_{bc}+\sqrt{-g} P^{abc}\delta g_{bc}\right)\ ,
\end{equation} we can calculate directly that $\omega(\phi, \delta \phi, \mathcal{L}_V \phi) = 0$ for the Killing vector field. For  a proof of this, see Appendix \ref{appendix:omega=0STEGR}. Taking this reduction of phase space into account and using the results from Section \ref{sec:firstlawBH}, the first law of black hole thermodynamics in STEGR will thus have the same form as in GR   \begin{equation}
     \frac{\kappa}{2\pi} \delta S = \delta E -   \Omega_H^{(\mu)} \delta \mathcal{J}_{(\mu)}.
\end{equation}

\vspace{0.5cm}

\subsection{Energy in STEGR}

Above we have defined the energy, which appears in the first law, from the Noether charge $\bm{Q}_N$ and the boundary term $\bm{\Theta}$ by \begin{equation}
    E = \int_{\infty} (\mathbf{Q}_N[t] - i_{t} \mathbf{B}), 
\end{equation} with $t^a$ an aymptotic time translation vector field and where $\bm{B}$ is defined by \begin{equation}
    \delta \int_{\infty}( i_{V} \mathbf{B}) = \int_{\infty} (i_{V} \bm{\Theta}).
\end{equation} For STEGR we have the boundary tensor (see Remark \ref{remark:boundary_terms_STEGR})
\begin{equation}
    \Theta^a=\frac{1}{2}\left(-g^{ac}\delta\Gamma^b_{bc}+g^{bc}\delta\Gamma^a_{bc}\right)+P^{abc}\delta g_{bc}
\end{equation}
leading to the Noether charge (for a calculation, see Appendix \ref{appendix:Noethercharge_STEGR})
\begin{equation}
    (\bm{Q}_N)_{fg} = \frac{1}{2}\bm{\epsilon}_{abfg}\bigg( -\mathcal{D}^{a}V^b + \delta^{[b}_d (Q^{a]}-\tilde{Q}^{a]})V^d \bigg).
\end{equation} To calculate the energy we will take $V = t^a$. Assuming the variations $\delta g$ and $\delta \Gamma$ vanish at infinity, we can take $\bm{B} = \bm{0}$, so that \begin{equation}
    E = \int_\infty \bm{Q}_N[t].
\end{equation} From the expression of $\bm{Q}_N$, we recognise in the first term exactly one-half of the Komar mass (see also equation (86) in \cite{Wald1994}). Thus for the energy we get \begin{equation} \label{eq:energySTEGR}
    \begin{split} 
        E &= \int_\infty \frac{1}{2}\bm{\epsilon}_{abfg}\bigg( -\mathcal{D}^{a}t^b + \delta^{[b}_d (Q^{a]}-\tilde{Q}^{a]})t^d \bigg),
    \end{split} 
\end{equation} which becomes the following in an asympotically flat vacuum spacetime (for a calculation see Appendix \ref{appendix:calc_energy_integral}) \begin{equation} \label{eq:energySTEGR2}
    \begin{split} 
        E = \int_{\infty} \mathrm{d}S \bigg( -\frac{1}{2}\partial_r g_{tt} + \frac{1}{2}\partial_t g_{rt} - \frac{1}{2}(Q^{r}-\tilde{Q}^{r}) \bigg),
    \end{split} 
\end{equation} which is to be evaluated on the equations of motion: that is for a torsion-free and curvature-free connection $\Gamma$ with $g, \Gamma$ satisfying the metric equation of motion. Recall that the equation of motion from varying the connection can be used to put a form of $\lambda_R$ and $\lambda_T$, without restricting $g$ or $\Gamma$. The remaining equations of motion for $g$ and $\Gamma$ in STEGR can be put in the form \begin{equation} \label{eq:STEGR_eom_g_R0}
    \begin{split} 
        \frac{2}{\sqrt{-g}} \nabla_{a}\left(\sqrt{-g} P^{a}{ }_{b c}\right)-q_{b c}-\mathbb{Q} g_{b c}=\mathfrak{T}_{b c}, \\ {\mathcal{R}^a}_{dbc}(g) = 2\nabla_{[b}{L^a}_{c]d}+ 2{L^a}_{[b|\lambda|}{L^\lambda}_{c]d},
    \end{split}
\end{equation} where the connection $\Gamma$ is taken to be symmetric, $\mathfrak{T}_{b c}=-\frac{2}{\sqrt{-g}} \frac{\delta \mathcal{S}_{M}}{\delta g^{b c}}$ is the matter energy-momentum and $q_{bc}$ stands for \begin{equation} \begin{split}
q_{b c} =-\frac{1}{4}\left(2 Q_{a d b} Q^{a d}{ }_{c}-Q_{b a d} Q_{c}{}^{a d}\right) +\frac{1}{2} Q_{a d b} Q^{d a}{ }_{c} \\ +\frac{1}{4}\left(2 Q_{a} Q^{a}{ }_{b c}-Q_{b} Q_{c}\right) -\frac{1}{2} \tilde{Q}_{a} Q^{a}{ }_{b c}. \end{split}
\end{equation} The first line in \eqref{eq:STEGR_eom_g_R0} is the equation of motion obtained from varying the metric (see equation (20) in \cite{HeisenbergL18}), and the second equation is the zero curvature constraint on $\Gamma$ (see also equation (25) in \cite{Heisenberg18}). Here $L$ is related to $Q$ by \eqref{eq:relateLQ}.

\vspace{0.5cm}

There is one case in which the expression for the energy for a vacuum asymptotically flat spacetime as in \eqref{eq:energySTEGR} becomes much simpler immediately, namely when choosing the coincident gauge $\Gamma = 0$. In this case, we get the following expression for the energy (for a calculation see Appendix \ref{appendix:calc_energy_integral}): \begin{equation}
    E = \frac{1}{2}\int_{\infty} \mathrm{d}S h^{i j}\left(\partial_{i} h_{r j}-\partial_{r} h_{i j}\right) = M_{ADM},
\end{equation} which we recognize as the ADM mass from General Relativity (see also equation (89) in \cite{Wald1990}). Here $h_{ij} = g_{ij}$ denotes the spatial metric, with $i,j$ spatial indices. 

\vspace{0.5cm}

\section{Results, discussion and further research} \label{sec:summary}
The geometrical formulation of gravity allows a multitude of interpretations and there is no unique manifestation of GR. The same theory of gravity can be either interpreted as the curvature of space-time, as Einstein did, or as a manifestation of the torsion or the non-metricity tensor. The latter allows a simpler formulation of GR purged from boundary terms. 

In this work we have applied Wald's Noether charge method for calculating the entropy of black holes to gravity theories with a general independent connection. We have focused here on the computation of the entropy in STEGR, but have introduced appropriate generalisations to theories with general connections, the so called Palatini theories. 
Wald's method for calculating black hole entropy relies solely on the presence of diffeomorphism invariance, which leads to a current and charge. Introducing diffeomorphisms by a Killing vector field, the entropy can be defined as the integral of the charge at the bifurcation surface.
As mentioned in Section \ref{sec:firstlawBH}, there seems to be an ambiguity in the entropy (which we named the $\bm{Y}$-ambiguity) coming from an ambiguity in the definition of the boundary term, namely the freedom to add some exact form. In STEGR however, the connection is constrained, imposed by Lagrange multipliers in the action and a reduction in phase space can be performed. After this reduction in phase space, this $\bm{Y}$-ambiguity is removed. When working in this reductionized phase space, we get the exact same expression for the entropy in STEGR as in GR. Furthermore, we obtained that the first law of black hole thermodynamics also holds in STEGR, and gave an explicit expression for the energy appearing in this law in STEGR as well. When working in the coincident gauge in STEGR, we evaluated the energy explicitly in an asymptotically flat vacuum spacetime and found that it equals the ADM mass, as in GR. 

\vspace{0.5cm}

Further research could involve applying this generalised Noether charge method to other Palatini theories of gravity with arbitrary connections and even generalized gravity theories beyond GR to find the entropy, such as $f(Q)$ theories or space-times endowed with curvature and torsion.  

\vspace{0.5cm}

\section*{Acknowledgements}
We would like to thank R. Wald for useful discussions. LH is supported by funding from the European Research Council (ERC) under the European Unions Horizon 2020 research and innovation programme grant agreement No 801781 and by the Swiss National Science Foundation grant 179740.

\vspace{1cm}


\printnomenclature


\begin{appendices}
\section{Boundary term in STEGR leading to same entropy as in GR} \label{appendix:boundarytermSTEGR_same_entropy_GR}

\begin{lemma} \label{lemma:deltagammaEHentropy}
A boundary term $\Theta^a$ as defined in \eqref{eq:defTheta} of the form \begin{equation} \label{eq:boundarydeltagammaentropy}
    \frac{1}{2}(-g^{ae}\delta \Gamma^{h}_{he} + g^{eh}\delta \Gamma^a_{eh})
\end{equation}  will lead to the same entropy as in GR from the Einstein-Hilbert action if the connection is curvature-free and torsion-free.
\end{lemma}
\begin{proof}
Now let us look at the boundary term as in \eqref{eq:boundarydeltagammaentropy} when the variation of $\Gamma$ induced by a vector field $V$. As stated in Appendix A of \cite{Heisenberg2020} (with a different sign convention), such a variation of $\Gamma$, for zero torsion and zero curvature, is given by \begin{equation}
    \mathcal{L}_{V} \Gamma^a_{bc} = \nabla_b \nabla_c V^a.
\end{equation} Using this in \eqref{eq:boundarydeltagammaentropy}, we get \begin{equation} \label{eq:deltagamma2}
    \begin{split}
         -g^{ae}\delta \Gamma^{h}_{he} + g^{eh}\delta \Gamma^a_{eh} &= g^{ae} \nabla_h \nabla_e V^h + g^{eh}\nabla_e \nabla_h V^a \\ &=  -g^{ae} \nabla_h \nabla_e (g^{hf} V_f) + g^{eh}\nabla_e \nabla_h (g^{af} V_f) \\ &= -g^{ae} g^{hf} \nabla_h \nabla_e V_f + g^{eh} g^{af} \nabla_h \nabla_e  V_f + V(\ldots) + \nabla V (\ldots) \\ &= g^{ae} g^{fh} (-\nabla_h \nabla_e V_f + \nabla_f \nabla_h V_e ) + V(\ldots) + \nabla V (\ldots) \\ &= g^{ae} g^{fh} (\mathcal{D}_f \delta g_{eh}-\mathcal{D}_e \delta g_{fh}) +V(\ldots) + \nabla V (\ldots).
    \end{split} 
    \end{equation} 
Here we have separated the important contributions for when $V$ is the Killing vector field, to be evaluated at the bifurcation surface, from the contributions which will give zero when evaluated in the Noether charge at the bifurcation surface, see also Remark \ref{remark:importantboundary}. Furthermore we have used that $\mathcal{L}_V g_{ab} = 2 \mathcal{D}_{(a} V_{b)}$. As calculated on p. 12 of \cite{Wald1994}, the boundary term for GR from the Einstein-Hilbert action with Lagrangian density $\mathcal{L} = \frac{1}{2}\sqrt{-g}\mathcal{R}(g)$ is equal to $\frac{1}{2}g^{ae} g^{fh} (\mathcal{D}_f \delta g_{eh}-\mathcal{D}_e \delta g_{fh})$, thus the expression for the boundary term in \eqref{eq:boundarydeltagammaentropy} leads to the same entropy as in GR. 
\end{proof}
\begin{remark}
Actually the above proof can be extended to the case where the connection is only curvature free. If the connection is curvature free, the variation induced by a vector field $V$ is given by (see \cite{Heisenberg2020}) \begin{equation}
        \label{eq:vargammatorsion} \begin{split}
            \mathcal{L}_V \Gamma^a_{\;\;\;bc }=\nabla_{b} \nabla_{c} V^a-T_{c d}^{a} \nabla_{b} V^{d}-(\nabla_{b} T_{c d}^{a}) V^{d},
        \end{split}
    \end{equation} which is the same as when the connection is also torsion-free, up to terms $\nabla V$ and $V$, which will not contribute to the entropy at the bifurcation surface.
\end{remark}

\vspace{0.5cm}

\section{Proof of equation \ref{eq:omega=0} in STEGR} \label{appendix:omega=0STEGR}
\begin{lemma} \label{lemma:omega}
When ignoring the $\mathbf{Y}$-ambiguities and working with the STEGR boundary term (as given in Remark \ref{remark:boundary_terms_STEGR}) \begin{equation}
    \Theta^a=\sqrt{-g}\left(-g^{c[a}\delta^{b]}_d\delta{\Gamma^d}_{bc}+P^{abc}\delta g_{bc}\right)\ ,
\end{equation} we have $\omega(\phi,\delta \phi, \mathcal{L}_V \phi) = 0$, where $\phi$ denote the dynamical fields $(g,\Gamma)$ and $V$ denotes the Killing vector field. 
\end{lemma}
\begin{proof}
By definition of $\omega$ (see Definition \ref{def:omega}), the boundary term 
\begin{equation}
    \Theta^a=\sqrt{-g}\left(-g^{c[a}\delta^{b]}_d\delta{\Gamma^d}_{bc}+P^{abc}\delta g_{bc}\right)\ , \end{equation}
leads to
\begin{equation}
    \omega^a=\Lie(\sqrt{-g}g^{c[a}\delta^{b]}_d)\delta{\Gamma^d}_{bc}-\Lie(\sqrt{-g}P^{abc})\delta g_{bc}-\delta(\sqrt{-g}g^{c[a}\delta^{b]}_d)\Lie{\Gamma^d}_{bc}+\delta(\sqrt{-g}P^{abc})\Lie g_{bc}\ .
\end{equation} 
Now we set $\Lie g_{ab}=0$, which leaves
\begin{equation}
\omega^a=-\sqrt{-g}\delta g_{ef}\left(\frac{1}{2}g^{ef} g^{c[a}\delta^{b]}_d\Lie{\Gamma^d}_{bc}-g^{e c} g^{f [a}\delta^{b]}_d\Lie{\Gamma^d}_{bc}+\Lie P^{aef}\right)\ ,    
\end{equation}
with the first term coming from $\delta\sqrt{-g}$. Since $\Lie g_{ab}=0$ we can then pull the $\Lie$ in front of the bracket. We then split the connection $\Gamma=\{\}+L$ as in \eqref{eq:connectiondecomp}, and focus on the $L$ part first. The contribution inside the bracket is
\begin{equation} \begin{split}
    \frac{1}{2}&g^{ef}g^{c[a}{L^{b]}}_{bc}-g^{e c} g^{f[a} {L^{b]}}_{bc}+P^{aef}\\&=\frac{1}{8}g^{ef}\left(-2Q^a+2\tilde Q^a\right)- \frac{1}{4}\left(- g^{a(f} Q^{e)}-Q^{aef}+2Q^{e af}\right)\\
    &+\frac{1}{4}\left(-Q^{aef}+2Q^{e af}+g^{ef}(Q^a-\tilde Q^a)-g^{a(e}Q^{f)}\right)\\
    &=0\ ,
\end{split}
\end{equation}
where we used the symmetry in $e$ and $f$ from $\delta g_{ef}$ in $\omega^a$, and used explicit expressions for $L$ and $P$ in STEGR as given in \eqref{eq:relateLQ} and \eqref{eq:Pbca}. This leaves
\begin{equation}
    \omega^a=-\sqrt{-g}\delta g_{ef}\mathcal{L}_V \left( \frac{1}{2}g^{ef}g^{c[a}\delta^{b]}_d \connsym{d}{bc}-g^{e c}g^{f[a}\delta^{b]}_d\connsym{d}{bc}\right)=0\ ,
\end{equation}
which vanishes on the account of $\Lie g_{ab}=0$.
\end{proof}

\vspace{0.5cm}

\section{Noether charge calculation in STEGR}
\label{appendix:Noethercharge_STEGR}
In this section we derive the Noether charge in STEGR using the algorithm of Lemma 1 in \cite{Wald1990b} from the boundary term for STEGR. Let us first explain the Noether charge and the algorithm from Lemma 1 in \cite{Wald1990b}. As described above, from the boundary term $\bm{\Theta}(\phi,\delta \phi)$ (obtained from varying the Lagrangian) we define the Noether current $\bm{J}$ (an $(n-1)$-form) by \begin{equation}
    \mathbf{J} = \bm{\Theta}(\phi,\mathcal{L}_V \phi) - i_{V} \mathbf{L},
\end{equation} which is closed on-shell, as proved in Lemma \ref{lemma:J_closed}, for all vector fields $V$. Thus as described in Lemma 1 in \cite{Wald1990b} there exists an $(n-2)$-form $\bm{Q}_N$ such that $\mathrm{d}\bm{Q}_N = \bm{J}$ on-shell; this $(n-2)$-form is called the Noether charge. In that same lemma the following algorithm for obtaining $\mathbf{Q}_N$ is given. The closed $p$-forms that are considered have the following form: \begin{equation}
    \alpha_{a_{1} \cdots \alpha_{p}}=\sum_{i=0}^{k} A^{(i)}{}_{a_{1} \cdots a_{p}}{}^{b_{1} \cdots b_{i}}{}_{\mu} \nabla_{\left(b_{1}\right.} \cdots \nabla_{\left.b_{i}\right)} V^{\mu},
\end{equation} where it is assumed that the coefficients $A$ satisfy \begin{equation}
    A^{(i)}{ }_{a_{1} \cdots a_{p}}{}^{b_{1} \cdots b_{i}}{ }_{\mu}=A^{(i)}{}_{[a_{1} \cdots a_p]}{ }^{(b_{1} \cdots b_{j})}{ }_{\mu}.
\end{equation} Here $\nabla$ is any (fixed) affine connection. As $\alpha$ is closed we have that \begin{equation}
    \nabla_{[c} \alpha_{a_1 \ldots a_p]} = 0.
\end{equation} One first looks at the term in $\alpha$ with the highest number of derivatives, and defines the $(p-1)$-form $\tau$ by \begin{equation} \label{eq:algorithm}
    \tau_{a_{2} \cdots a_{p}}=\frac{m}{n-p+m} A^{(m)}{ }_{c a_{2} \cdots a_{p}}{}^{c b_{2} \cdots b_{m}}{ }_{\mu} \boldsymbol{\nabla}_{b_{2}} \cdots \boldsymbol{\nabla}_{b_{m}} V^{\mu} .
\end{equation} Then $\mathrm{d}\tau$ and $\alpha$ will have the same highest term derivative (namely the $m$-th symmetrized derivative). Here $n$ is the spacetime dimension in our case. One temporarily sets the Noether charge $\bm{Q}_N = \tau$, which will be updated along the way. Next one defines $\alpha^{(1)} = \alpha - \mathrm{d}\tau$, which is closed but has a higest derivative of $V$ of order $m-1$. One applies the same procedure, obtaining a $\tau^{(1)}$, and updates the `final' Noether charge by $\mathbf{Q}_N = \tau + \tau^{(1)}$. One repeats this procedure until the obtained $\alpha^{(k)}$ does not contain any derivatives of $V$ anymore; in that case one should have $\alpha^{(k)}$ is zero, as the only identically closed $p$-forms that depend linear on $V$ and not on its derivatives is the trivial zero-form (see argument below equation (11) in \cite{Wald1990}). Note that this procedure implies that for constructing $\mathbf{Q}_N$ we should only look at terms with derivatives in $V$; the terms linear in $V$ such as $i_V L$ are not used in the construction in the algorithm (but they are important in order for $\bm{J}$ to be closed on-shell of course). 

\vspace{0.5cm} 

Let us apply this algorithm to find the Noether charge in the Palatini formalism with a torsion-free connection, whose result was used in Lemma \ref{lemma:formNoerthercharge_dynconnection}. For simplicity of notation we take our spacetime to have dimension $4$, but the calculation is analogous. The current is given by (see \eqref{eq:J_symmetricPalatini}) \begin{equation} \begin{split}
    (\textbf{J})_{efg} &=\varepsilon_{a efg}2{E_{Rd}}^{cab}\nabla_b\nabla_c V^d+\text{terms linear in $V$ and $\nabla V$} \\ &= \varepsilon_{aefg}(E_{Rd}{}^{cab} + E_{Rd}{}^{bac})\nabla_{(b}\nabla_{c)}V^d + \ldots, 
\end{split}
\end{equation} leading to a Noether charge (only keeping track of the highest derivative terms in $\nabla V$):
\begin{equation}
\begin{split}
      (\bm{Q}_N)_{fg} &= \frac{2}{3} \varepsilon_{abfg}(E_{Rd}{}^{cab} + E_{Rd}{}^{bac})\nabla_c V^d + \ldots \\ &= \frac{2}{3} \varepsilon_{abfg}(E_{R}^{dcab} + E_{R}^{dbac})\nabla_c V_d + \ldots \\ &=  \frac{2}{3} \varepsilon_{abfg}(E_{R}^{dc[ab]} + \frac{E_{R}^{dbac}-E_R^{dabc}}{2})\nabla_c V_d  + \ldots \\ &=  \frac{2}{3} \varepsilon_{abfg}(E_{R}^{dc[ab]} + \frac{-E_{R}^{dacb}-E_R^{dcba}-E_R^{dabc}}{2})\nabla_c V_d  + \ldots \\ &= \frac{2}{3} \varepsilon_{abfg}(\frac{3}{2}E_{R}^{dc[ab]})\nabla_c V_d  + \ldots
      \\ &= \varepsilon_{abfg} E_R^{dcab} \nabla_c V_d + \ldots,
\end{split} 
\end{equation} where we used the antisymmetry of $\bm{\epsilon}_{abfg}$ in $a,b$ and we used the symmetries of $E_R^{abcd}$ inherited from $R^{abcd}$, namely \begin{equation}
    R^{abcd} = - R^{abdc}
\end{equation} and the first Bianchi identity (in the absence of torsion!) \begin{equation}
    R^{abcd} + R^{acdb} + R^{adbc} = 0.
\end{equation} Note that indeed for a general spacetime dimension the above prefactor of $\bm{Q}_N$ will be the same; the $m,n,p$ as in \eqref{eq:algorithm} are then the following: we are dealing with a spacetime of dimension $n$, with a $p$-form with $p = n-1$ and $\bm{J}$ has derivatives of $V$ of highest order $m = 2$, leading to a prefactor \begin{equation}
    \frac{m}{n-p+m} = \frac{2}{n-(n-1)+2} = \frac{2}{3}.
\end{equation}

\vspace{0.5cm}

Now let us apply this algorithm to calculate the Noether charge in STEGR. As explained in Remark \ref{remark:boundary_terms_STEGR} the boundary term in STEGR has the following form (see Remark \ref{remark:boundary_terms_STEGR}): \begin{equation}
    \bm{\Theta}_{efg}(\phi,\delta \phi) = \bm{\epsilon}_{aefg} \bigg( \frac{1}{2}\left(-g^{ac}\delta\Gamma^b_{bc}+g^{bc}\delta\Gamma^a_{bc}\right)+P^{abc}\delta g_{bc} \bigg),
\end{equation} where $\bm{\epsilon}$ is the totally antisymmetric Levi-Civita tensor, so that \begin{equation}
\bm{\Theta}_{efg}(\phi,\mathcal{L}_V \phi) = \bm{\epsilon}_{aefg} \bigg( \frac{1}{2}(g^{bc}\delta^a_d-g^{ac}\delta^b_d) \nabla_b \nabla_c V^d+2P^{abc} \mathcal{D}_b V_c \bigg),
\end{equation} where we used $\mathcal{L}_V \Gamma^{a}_{bc} = \nabla_b \nabla_c V^a$ and $\mathcal{L}_V g_{ab} = 2\mathcal{D}_{(a} V_{b)}$ \footnotetext{Note that this is already with constrained connection, i.e. zero torsion and zero curvature.}. As the connection is torsion- and curvature-free, we can replace $\nabla_b \nabla_c$ by $\nabla_{(b} \nabla_{c)}$ in the above expression, leading to  \begin{equation}
\bm{\Theta}_{efg}(\phi,\mathcal{L}_V \phi) = \bm{\epsilon}_{aefg} \bigg( \frac{1}{2}(g^{bc}\delta^a_d-g^{ac}\delta^b_d) \nabla_{(b} \nabla_{c)} V^d+2P^{abc} \mathcal{D}_b V_c \bigg).
\end{equation} As explained in \cite{Wald1990b}, before applying the algorithm of Lemma 1 in \cite{Wald1990b} to the above boundary terms we should still make sure that the term multiplying $\nabla_{(b} \nabla_{c)}$ is symmetrized over $b$ and $c$:  \begin{equation}
\bm{\Theta}_{efg}(\phi,\mathcal{L}_V \phi) = \bm{\epsilon}_{aefg} \bigg( \frac{1}{2}(g^{bc}\delta^a_d-g^{a(b}\delta^{c)}_d) \nabla_{(b} \nabla_{c)} V^d+2P^{abc}\mathcal{D}_b V_c \bigg).
\end{equation} Before we proceed we will split the covariant derivative $\nabla$ into the Levi-Civita derivative and a non-metricity part. The advantage of this is that for a $2$-form $\tau$ corresponding to the highest derivative term of $\bm{\Theta}$ we will not have to do any more computations to get $\mathrm{d}\tau$. Alternatively, the calculation can proceed by just working with $\nabla$, but then in calculating $\mathrm{d}\tau$ we will pick up terms contributing to lower derivative terms. Thus we will first perform the following substitution: \begin{equation}
    \begin{split}
        \nabla_{(b} \nabla_{c)} V^d &= \nabla_{(b} \mathcal{D}_{c)}V^d + \nabla_{(b} L^d{}_{c)e}V^e \\ &= \mathcal{D}_{(b} \mathcal{D}_{c)}V^d-L^{e}{}_{bc}\mathcal{D}_e V^d + L^{d}{}_{e (b} \mathcal{D}_{c)} V^e + L^d{}_{e(b} \mathcal{D}_{c)} V^e \\ &+ (\mathcal{D}_{(b} L^d{}_{c)e}-L^{m}{}_{bc}L^d{}_{me} - L^d{}_{m(b}L^m{}_{c)e})V^e.
    \end{split}
\end{equation} In this way we get \begin{equation}
    \begin{split}
        \bm{\Theta}^{(2)}_{efg} = \bm{\epsilon}_{aefg} \frac{1}{2} \bigg(g^{bc}\delta^a_d-g^{a(b}\delta^{c)}_d \bigg) \mathcal{D}_{(b} \mathcal{D}_{c)} V^d,
    \end{split}
\end{equation} for the second-order terms (here $\bm{\Theta}^{(i)}$ denotes the term in $\bm{\Theta}$ with only derivatives of $V$ of order $i$) and for the first-order terms we have 
\begin{equation}
    \begin{split}
        \bm{\Theta}^{(1)}_{efg} &= \bm{\epsilon}_{aefg} \bigg(2P^{ab}{}_d + L^{ab}{}_d - \frac{1}{2}L^{c}{}_{cd}g^{ab}-\frac{1}{2}L^{bc}{}_c \delta^a_d \bigg) \mathcal{D}_b V^d \\ &= \bm{\epsilon}_{aefg} \bigg( \frac{1}{2} \delta^b_d Q^a-\frac{1}{2} \delta^a_d Q^b-\frac{1}{2} \delta^b_d \tilde{Q}^a+\frac{1}{2} \delta^a_d \tilde{Q}^b \bigg) \mathcal{D}_b V^d \\ &= \bm{\epsilon}_{aefg}\bigg(\delta^{[b}_d (Q^{a]}-\tilde{Q}^{a]}) \bigg) \mathcal{D}_b V^d,
    \end{split}
\end{equation} where we used the expressions of $P$ and $L$ in terms of $Q$ as given in \eqref{eq:Pbca} and \eqref{eq:relateLQ}. Following the algorithm of Lemma 1 of \cite{Wald1990b} as described above, we arrive at a Noether charge \begin{equation}  \label{eq:Noethercharge_STEGR} \begin{split}
    (\bm{Q}_N)_{fg} = \frac{1}{2}\bm{\epsilon}_{abfg}\bigg( -\mathcal{D}^{a}V^b + \delta^{[b}_d (Q^{a]}-\tilde{Q}^{a]})V^d \bigg).
    \end{split}
\end{equation}
For completeness let us write down a detailed calculation: for the highest-order term to obtain the corresponding part in the Noether charge  \begin{equation}
    \begin{split}
        \bm{\Theta}^{(2)}_{efg} &= \frac{1}{2} \bm{\epsilon}_{aefg}(g^{bc}\delta^a_d-\frac{1}{2}g^{ac}\delta^b_d-\frac{1}{2}g^{ab}\delta^c_d) \mathcal{D}_{(b} \mathcal{D}_{c)}V^d \\ \rightarrow \bm{Q}_{fg} &= \frac{1}{3} \bm{\epsilon}_{abfg}(g^{bc}\delta^a_d - \frac{1}{2}g^{ac} \delta^b_d - \frac{1}{2}g^{ab}\delta^c_d)\mathcal{D}_c V^d \\ &=  \frac{1}{3} \bm{\epsilon}_{abfg}(g^{bc}\delta^a_d - \frac{1}{2}g^{ac} \delta^b_d )\mathcal{D}_c V^d \\ &= \frac{1}{3} \bm{\epsilon}_{abfg}(\frac{3}{2} g^{c[b}\mathcal{D}_c V^{a]}) \\ &= -\frac{1}{2} \bm{\epsilon}_{abfg} \mathcal{D}^{[a} V^{b]},
    \end{split}
\end{equation} where we used the antisymmetry of $\bm{\epsilon}_{abfg}$ in $a$ and $b$ several times, while for the first-order term we have \begin{equation}
    \begin{split}
        \bm{\Theta}^{(1)}_{efg} &= \bm{\epsilon}_{aefg}\bigg(\delta^{[b}_d (Q^{a]}-\tilde{Q}^{a]}) \bigg) \mathcal{D}_b V^d \\ &\rightarrow \bm{Q}_{fg} = \frac{1}{2}\bm{\epsilon}_{abfg} \bigg(\delta^{[b}_d (Q^{a]}-\tilde{Q}^{a]}) \bigg) V^d,
    \end{split}
\end{equation} leading to \eqref{eq:Noethercharge_STEGR}. Note that from this explicit expression for the Noether charge one can also see that the entropy in STEGR will be the same as in GR. 

\vspace{0.5cm}

\section{Calculation of energy from Noether charge integral equation \ref{eq:energySTEGR}}
\label{appendix:calc_energy_integral}
Here we give a more detailed calculation of the energy as the integral of the Noether charge over a spherical surface at infinity in vacuum, as given in \eqref{eq:energySTEGR}, based on p.13 of \cite{Wald1994}. We consider an asymptotically flat spacetime such that there exists a flat metric so that  \begin{equation} \label{eq:energy_metric_leadingorder} \begin{split}
 g_{\mu \nu}=\eta_{\mu \nu}+O(1 / r), \\
    \frac{\partial g_{\mu \nu}}{\partial x^{\alpha}}=O\left(1 / r^{2}\right)
\end{split}
\end{equation} in a global inertial coordinate system of $n_{ab}$. We denote the asymptotic time translation $t^a = (\partial_t)^a$ and we integrate over the 2-sphere at infinity being the limit $r\rightarrow \infty$ of the spheres $r,t =$ constant. We have \begin{equation} 
    \begin{split} 
        E &= \int_\infty \frac{1}{2}\bm{\epsilon}_{abfg}\bigg( -\mathcal{D}^{a}t^b + \delta^{[b}_d (Q^{a]}-\tilde{Q}^{a]})t^d \bigg)  \\ &= \int_\infty \sin(\theta) (d t \wedge d r \wedge d \theta \wedge d \phi)_{abcd} \bigg( \frac{1}{2} \mathcal{D}^{[b} t^{a]} +  
    \frac{1}{2}\delta^{[b}_t (Q^{a]}-\tilde{Q}^{a]}) \bigg) \\ &=  \int_\infty \sin(\theta) (d \theta \wedge d \phi)_{cd}\bigg( \mathcal{D}^{[r} (\partial_t)^{t]} +
    \delta^{[r}_t (Q^{t]}-\tilde{Q}^{t]}) \bigg) \\ &= \int_\infty \sin(\theta) (d \theta \wedge d \phi)_{cd}\bigg( \mathcal{D}^{[r} (\partial_t)^{t]} -
    \frac{1}{2}(Q^{r}-\tilde{Q}^{r}) \bigg).
    \end{split}
\end{equation} Let us now focus on the first term, which we also recognise as one-half of Komar mass. The result is also given on p. 15 of \cite{Wald1994}, but for completeness we perform the calculation here: \begin{equation} \begin{split}
     \int_{\infty} \mathrm{d}S \mathcal{D}^{[r} (\partial_t)^{t]} &=  \int_{\infty} \mathrm{d}S\frac{1}{2} \bigg(g^{rk} \{ \substack{t\\kt} \} - g^{tk} \{ \substack{r\\kt} \} \bigg) =  \int_{\infty} \mathrm{d}S\frac{1}{2} \bigg( \{ \substack{t\\rt} \} + \{ \substack{r\\tt} \} \bigg) \\ &=  \int_{\infty} \mathrm{d}S\frac{1}{2} \bigg( \frac{g^{tt}}{2} \partial_r g_{tt} + \frac{g^{rr}}{2}(2 \partial_t g_{rt} - \partial_r g_{tt}) \bigg) \\ &= \int_{\infty} \mathrm{d}S \bigg( -\frac{1}{2}\partial_r g_{tt} + \frac{1}{2}\partial_t g_{rt} \bigg),
\end{split}
\end{equation} where $\mathrm{d} S = \sin(\theta) (d \theta \wedge d \phi)_{cd}$ is the induced volume form on the sphere at infinity. Thus we finally get
\begin{equation}
    \begin{split}
        E = \int_{\infty} \mathrm{d}S \bigg( -\frac{1}{2}\partial_r g_{tt} + \frac{1}{2}\partial_t g_{rt} - \frac{1}{2}(Q^{r}-\tilde{Q}^{r}) \bigg).
    \end{split}
\end{equation} Next, let us see what this gives us in the coincident gauge, that is with $\Gamma^{a}_{bc} = 0$. In this gauge we have (considering only the leading order in $r^n$ terms (from \eqref{eq:energy_metric_leadingorder}) \begin{equation}
    Q_r-\widetilde{Q}_r = g^{fh}(\partial_r g_{fh}-\partial_f g_{rh}) = g^{tt}(\partial_r g_{tt}-\partial_t g_{rt}) + g^{ij} (\partial_r g_{ij} - \partial_i g_{rj}),
\end{equation} where $i,j$ denote spatial indices. Thus for the energy we get \begin{equation}
    E = \frac{1}{2} \int_{\infty} \mathrm{d}S g^{ij}(\partial_i g_{rj}-\partial_r g_{ij}) = M_{ADM},
\end{equation} which is equal to the ADM mass ($M_{ADM}$ denoted here) that we know from general relativity. 

\vspace{0.5cm}

\section{Boundary terms for Palatini with connection and non-zero torsion} \label{appendix:boundary_Palatini+torsion}

In Section \ref{sec:Palatini_BT_sympl} we have obtained the boundary terms in a Palatini formalism where the connection was torsion-free. Here we extend this to obtain boundary terms in a Palatini setting with dynamical metric $g$ and an independent connection $\Gamma$ (with covariant derivative $\nabla$), which can have non-zero torsion $T$ as well. We will use the Lagrangian density $\mathcal{L}$ related to the Lagrangian $n$-form by $\bm{L} = \varepsilon \mathcal{L}$. The key difference is that now in $\mathcal{L}$ a term like $\nabla_a A^a$ is not a pure boundary term, as we have \begin{equation}
    \nabla_a \alpha^a = \partial_a \alpha^a + T^a_{ab} \alpha^b,
\end{equation} for a tensor density $\alpha^a$ of weight $-1$, as explained in Appendix \ref{appendix:BT}. The general gravity action is given by
\begin{equation} \label{eq:Palatiniaction}
 \begin{split}
    S=\int d^4x \mathcal{L}\bigg(& g_{ab}, \text{ } {R^a}_{bcd}, \nabla_{a_1}R^{a}_{bcd},\text{ } \ldots, \nabla_{(a_1}\ldots \nabla_{a_k)}R^{a}_{bcd}, \text{ }  Q_{abc},  \nabla_{a_1} Q_{abc},\text{ } \ldots, \\ & \nabla_{(a_1}\ldots \nabla_{a_l)}Q_{abc}, \text{ } T^{a}{}_{bc}, \nabla_{a_1} T^{a}{}_{bc}, \text{ } \ldots, \nabla_{(a_1} \ldots \nabla_{a_m)}T^{a}{}_{bc}, \text{ }  \Psi, \text{ } \lambda\bigg) ,
\end{split}
\end{equation}
with the Lagrangian $\mathcal{L}$ an arbitrary density of weight $-1$, depending on the metric, connection, Riemann, torsion, and non-metric tensor, as well as extra dynamical fields $\Psi$ (indices of $\Psi$ are suppressed, but kept arbitrary), for example matter fields, and Lagrange multipliers $\lambda$ (this means that $\mathcal{L}$ depends only linearly on the $\lambda$'s, and does not contain derivatives of these). We also allow the Lagrangian density to depend on symmetrized $\nabla$-derivatives of the matter fields $\Psi$, but for notational simplicity we just denote this by $\Psi$ in the density $\mathcal{L}$ above. Note that derivatives of the metric are given by the non-metricity, and that in general for any linear connection terms as $\nabla_{a_1} \ldots \nabla_{a_k}$ can expressed in terms of the symmetrized derivative $\nabla_{(a_1} \ldots \nabla_{a_k)}$ plus curvature terms, torsion terms and lower order derivatives \cite{Wald1990,carroll_2019}, which is why above we only considered symmetrized derivatives in $\mathcal{L}$.

\begin{lemma} \label{lemma:palatiniboundary+torsion}
We can write a general variation of the Lagrangian density $\mathcal{L}$, defined in \eqref{eq:Palatiniaction}, as
\begin{equation} \begin{split} 
    {\delta\mathcal{L}} &=(A^{ab}_g+(T^{d}_{dc}-\nabla_c) F_Q^{cab})\delta g_{ab}\\ &+(B_{\Gamma a}^{bc}+(T^e_{ed}-\nabla_d)2 {F_{Ra}}^{cdb}+{F_{Ra}}^{cde}{T^b}_{de}+2{F_{Ta}}^{[bc]}-2F_Q^{bc}{}_a)\delta {\Gamma^a}_{bc}\\
 &+F_\Psi\delta\Psi + C_\lambda \delta\lambda\\
&+\partial_d(2{F_{Ra}}^{cdb}\delta{\Gamma^a}_{bc}+F_Q^{dab}\delta g_{ab})+\partial_a{\tilde\Theta}^a\ ,
\end{split} \end{equation}
The first two lines give the equations of motion for the metric and connection, while the third line are the equations of motion for the extra fields $\Psi$ and the constraints. The fourth line gives the desired boundary term. Since the $F_Q^{dab}\delta g_{ab}$ part of the boundary term will only contain variations of the metric, it will ultimately not contribute to the entropy, and we thus absorb it in $\tilde\Theta$. The boundary term thus is
\begin{equation} 
    \Theta^d=2{F_{Ra}}^{cdb}\delta{\Gamma^a}_{bc}+\tilde\Theta^d\ ,
\end{equation}
where the contribution $\tilde{\boldsymbol{\Theta}}$ contains only terms like $\delta\nabla_{a_1}...\nabla_{a_i}R_{abcd}$, $\delta \nabla_{a_1} Q_{abc}$ and $\delta \Psi$ etc, e.g. it contains only variations but not derivatives of variations. The functions $F_i$ are similar to the equations of motion for the curvature, torsion, and non-metricity tensor, and extra fields $\Psi$, as if they were viewed as independent fields, up to torsion terms. For example
\begin{equation} \begin{split}
    F_R^{abcd}=&\frac{\partial\mathcal{L}}{\partial R_{abcd}}+(T^d_{da_1} -\nabla_{a_1})\frac{\partial\mathcal{L}}{\partial \nabla_{a_1}R_{abcd}}\\ &+(T^d_{da_2} -\nabla_{a_2})\bigg( (T^d_{da_1}-\nabla_{a_1})\frac{\partial\mathcal{L}}{\partial \nabla_{(a_1} \nabla_{a_2)} R_{abcd}} \bigg) + ...\ .
\end{split} \end{equation} 
\end{lemma}
\begin{proof}
Varying $\textbf{L} = \varepsilon \mathcal{L}$ gives
\begin{align}
\nonumber    \delta \textbf{L}&=\varepsilon\left(\frac{\partial \mathcal{L}}{\partial g_{ab}}\delta g_{ab}+\frac{\partial\mathcal{L}}{\partial R_{abcd}}\delta R_{abcd}+\frac{\partial\mathcal{L}}{\partial\nabla_{a_1}R_{abcd}}\delta\nabla_{a_1}R_{abcd}+...\right.\\
\nonumber\\
\nonumber    &+\frac{\partial\mathcal{L}}{\partial T_{abc}}\delta T_{abc}+\frac{\partial\mathcal{L}}{\partial\nabla_{a_1}T_{abc}}\delta\nabla_{a_1}T_{abc}+...\\
\nonumber\\
\nonumber &+\frac{\partial\mathcal{L}}{\partial Q_{abc}}\delta Q_{abc}+\frac{\partial\mathcal{L}}{\partial\nabla_{a_1}Q_{abc}}\delta\nabla_{a_1}Q_{abc}+...\\
\nonumber\\
 &\left.+\frac{\partial\mathcal{L}}{\partial \Psi}\delta \Psi+\frac{\partial\mathcal{L}}{\partial\nabla_{a_1}\Psi}\delta\nabla_{a_1}\Psi+...+\frac{\partial\mathcal{L}}{\partial \lambda}\delta\lambda\right)\ ,
\end{align}
where $...$ stands for variations with respect to higher-order derivatives of $R,T,Q$ or $\Psi$. The variation of $\sqrt{-g}$ is included in the first term. In order to treat the terms containing derivatives of Riemann and non-metricity tensors, we compute for any tensor $A$
\begin{equation}
\begin{split}
\frac{\partial\mathcal{L}}{\partial \nabla_{(a_1}...\nabla_{a_i)}A}&\delta \nabla_{(a_1}...\nabla_{a_i)} A \text{ } = \text{ }  \frac{\partial\mathcal{L}}{\partial \nabla_{(a_1}...\nabla_{a_i)}A} \nabla_{a_1}\delta \nabla_{a_2}...\nabla_{a_i}A + \text{terms with }\delta \Gamma\\
    =\nabla_{a_1}&\left(\frac{\partial\mathcal{L}}{\partial \nabla_{(a_1}...\nabla_{a_i)}A}\delta\nabla_{a_2}...\nabla_{a_i}A\right)+ \text{terms with }\delta \Gamma \\ &-\nabla_{a_1}\left(\frac{\partial\mathcal{L}}{\partial \nabla_{(a_1}...\nabla_{a_i)}A}\right)\delta \nabla_{a_2}...\nabla_{a_i}A \\ = \partial_{a_1}&\left(\frac{\partial\mathcal{L}}{\partial \nabla_{(a_1}...\nabla_{a_i)}A}\delta\nabla_{a_2}...\nabla_{a_i}A\right) + T^{d}_{d a_1} \left(\frac{\partial\mathcal{L}}{\partial \nabla_{(a_1}...\nabla_{a_i)}A}\delta\nabla_{a_2}...\nabla_{a_i}A\right)  \\ &+ \text{terms with }\delta \Gamma
     -\nabla_{a_1}\left(\frac{\partial\mathcal{L}}{\partial \nabla_{(a_1}...\nabla_{a_i)}A}\right)\delta \nabla_{a_2}...\nabla_{a_i}A \\ =  \partial_{a_1}&\left(\frac{\partial\mathcal{L}}{\partial \nabla_{(a_1}...\nabla_{a_i)}A}\delta\nabla_{a_2}...\nabla_{a_i}A\right)  + \text{terms with }\delta \Gamma \\ &+
     (T^d_{da_1}-\nabla_{a_1})\left(\frac{\partial\mathcal{L}}{\partial \nabla_{(a_1}...\nabla_{a_i)}A}\right)\delta \nabla_{a_2}...\nabla_{a_i}A \\ = \partial_{a_1}&\left(\frac{\partial\mathcal{L}}{\partial \nabla_{(a_1}...\nabla_{a_i)}A}\delta\nabla_{a_2}...\nabla_{a_i}A\right)  + \text{terms with }\delta \Gamma \\ &+
     \partial_{a_2}\bigg((T^d_{da_1}-\nabla_{a_1})\left(\frac{\partial\mathcal{L}}{\partial \nabla_{(a_1}...\nabla_{a_i)}A}\right)\delta \nabla_{a_3}...\nabla_{a_i}A \bigg) \\ &+ (T^d_{da_2}-\nabla_{a_2}) \bigg( (T^d_{da_1}-\nabla_{a_1})\left(\frac{\partial\mathcal{L}}{\partial \nabla_{(a_1}...\nabla_{a_i)}A}\right) \bigg) \delta \nabla_{a_3}...\nabla_{a_i}A  ,
\end{split}
\end{equation}
where the terms with $\delta\Gamma$ arise from the variation of the first covariant derivative $\nabla_{a_1}$ that we pulled out. Each of the terms containing $\delta{\Gamma^a}_{bc}$ will contribute to the connection equation of motion. The first and third term in the last expression will contribute to the boundary term. Iterating this procedure until all variations of derivatives are put in the boundary terms and performing this procedure for all derivative terms in $\delta\textbf{L}$, collecting all terms of each of the variations one finds
\begin{equation}
\label{eq:Palatinivariation}
    \delta\textbf{L}=A^{ab}_g \delta g_{ab}+B_{\Gamma a}^{bc} \delta \Gamma^a_{bc}+F_R^{abcd}\delta R_{abcd}+F_T{}_a{}^{bc}\delta T^a{}_{bc}+F_Q^{abc}\delta Q_{abc}+ F_\Psi\delta\Psi + C_\lambda \delta\lambda+d\tilde{\boldsymbol{\Theta}} ,
\end{equation} where the contribution $\tilde{\boldsymbol{\Theta}}$ contains only terms like $\delta\nabla_{a_1}...\nabla_{a_i}R_{abcd}$, $\delta \nabla_{a_1} Q_{abc}$ and $\delta \Psi$ etc, e.g. it contains only variations but not derivatives of variations. The functions $F_i$ are then similar to equations of motion for the curvature, torsion, and non-metricity tensor, and extra fields $\Psi$, as if they were viewed as independent fields, up to torsion factors; for example
\begin{equation} \begin{split}
    E_R^{abcd}=&\frac{\partial\mathcal{L}}{\partial R_{abcd}}+(T^d_{da_1} -\nabla_{a_1})\frac{\partial\mathcal{L}}{\partial \nabla_{a_1}R_{abcd}}\\ &+(T^d_{da_2} -\nabla_{a_2})\bigg( (T^d_{da_1}-\nabla_{a_1})\frac{\partial\mathcal{L}}{\partial \nabla_{(a_1} \nabla_{a_2)} R_{abcd}} \bigg) + ...\ .
\end{split} \end{equation} 
The functions $C_\lambda$ are the constraints coming from the Lagrange multipliers. Note that $A_g$ also contains the terms coming from variations of factors containing the determinant of the metric $\sqrt{-g}$, which we have not made explicit in \eqref{eq:Palatiniaction}. \\
Finally, we have that the variations appearing in \eqref{eq:Palatinivariation} can be writtten as variations of the metric and connection as
\begin{align}
    \delta {R^a}_{bcd}&=2\nabla_{[c}\delta{\Gamma^a}_{d]b} + {T^e}_{cd}\delta{\Gamma^a}_{eb}\\
    \delta {T^a}_{bc}&=2\delta {\Gamma^a}_{[bc]}\\
    \delta Q_{abc}&= \nabla_a\delta g_{bc}-2\delta {\Gamma^{d}}_{a(b}g_{c)d}\ .
\end{align}  The term $\delta T$ will thus only contribute to the connection equation of motion, as well as the $\delta\Gamma$ pieces from $\delta Q$ and $\delta R$. The only contributions to the boundary term come from $\delta R$ and the metric variation in $\delta Q$. Hence
\begin{equation} \label{eq:varL5.1} \begin{split}
{\delta\mathcal{L}} &=(A^{ab}_g+(T^{d}_{dc}-\nabla_c) F_Q^{cab})\delta g_{ab}\\ &+(B_{\Gamma a}^{bc} +(T^e_{ed}-\nabla_d)2 {F_{Ra}}^{cdb}+{F_{Ra}}^{cde}{T^b}_{de}+2{F_{Ta}}^{[bc]}-2F_Q^{bc}{}_a)\delta {\Gamma^a}_{bc}\\
 &+F_\Psi\delta\Psi + C_\lambda \delta\lambda\\
&+\partial_d(2{F_{Ra}}^{cdb}\delta{\Gamma^a}_{bc}+F_Q^{dab}\delta g_{ab})+\partial_a{\tilde\Theta}^a\ ,
\end{split} \end{equation}
where we used the symmetries of the $F_i$'s, inherited from the objects through which they are defined. The first and second line give the equations of motion for the metric and connection, while the third line are the equations of motion for the extra fields $\Psi$ and the constraints. The fourth line gives the desired boundary term. Since the $F_Q$ part of the boundary term only contains a variation of the metric, it will ultimately not contribute to the entropy (see Remark \ref{remark:importantboundary}), and we thus absorb it in $\tilde\Theta$. The boundary term then is
\begin{equation}
    \Theta^a=2{F_{Rd}}^{cab}\delta{\Gamma^d}_{bc}+\tilde\Theta^a\ ,
\end{equation}
or in form language
\begin{equation}
    (\boldsymbol{\Theta})_{a_2 \ldots a_n}=\varepsilon_{a a_2 \ldots a_n} 2{F_{Rd}}^{cab}\delta{\Gamma^d}_{bc}+(\tilde{\boldsymbol{\Theta}})_{a_2 \ldots a_n}\ .
\end{equation}

\end{proof}

\begin{remark}
Following the algorithm as explained in Appendix \ref{appendix:Noethercharge_STEGR}, the Noether charge in this setting will have the form (writing only the highest-order derivatives) \begin{equation} \begin{split}
    (\bm{Q}_N)_{fg} &= \frac{2}{3} \varepsilon_{abfg}(F_{R}^{dcab} + F_{R}^{dbac})\nabla_c V_d + \ldots.
    \end{split}
    \end{equation}
\end{remark}

\vspace{0.5cm}

\section{Hodge dual and boundary terms} \label{appendix:BT}

When considering variations of actions we will be interested in boundary terms, and in this process we will make use of the Hodge duality, so let us revise the basics of Hodge dualities. Let $M$ denote an $n$-dimensional manifold with metric $g$, then the hodge duality is a bijection map between $p$-forms and $(n-p)$-forms. Having an inner product on every tangent space $T_pM$ from the metric $g$, we can define an inner product on the $k$-forms on $\Omega^k(M)$ on $M$ in the following way. For two $k$-forms $\bm{\alpha}$ and $\bm{\beta}$ of the form $\bm{\alpha} = \bm{\alpha}_1 \wedge \bm{\alpha}_2 \wedge \ldots \wedge \bm{\alpha}_k, \bm{\beta} = \bm{\beta}_1 \wedge \bm{\beta}_2 \wedge \ldots \wedge \bm{\beta}_k$, the inner product between $\bm{\alpha}$ and $\bm{\beta}$ is defined by \begin{equation}
    \langle \bm{\alpha}, \bm{\beta} \rangle = \det\big{(} \langle \bm{\alpha}_i, \bm{\beta}_j \rangle \big{)}_{i,j = 1}^k,
\end{equation}  and can be extended to an inner product on all $k$-forms by linearity. On every tangent space $T_p M$ one can define the unit volume form $\bm{\omega}$ in terms of an orthonormal basis $e_1, e_2, \ldots, e_n$ by $\bm{\omega} = e_1 \wedge e_2 \wedge \ldots \wedge e_n$. On a manifold with metric we can define such a volume form or $n$-form  (which will be called the Levi-Civita tensor) also in the following way. Let $\varepsilon$ denote the totally antisymmetric symbol, then the Levi-Civita tensor is defined by $\bm{\epsilon} = \sqrt{-g}\varepsilon$ and is an $n$-form. The map from $p$-forms on $M$ to $(n-p)$-forms on $M$ is often called the Hodge star operator and is denoted by $\star$:\footnotetext{Here we have denoted the space of $k$-forms on $M$ by $\Omega^k(M)$.} \begin{equation} \begin{split}
    &\star: \Omega^p(M)  \rightarrow \Omega^{n-p}(M): \bm{\alpha} \mapsto \star \bm{\alpha}, \\ & \text{ defined by   }  \;\; \bm{\beta} \wedge (\star \bm{\alpha}) = \langle \bm{\beta}, \bm{\alpha} \rangle \bm{\epsilon} \text{ for all } \bm{\beta} \in \Omega^p(M)
\end{split}
\end{equation} 
Using the Levi-Civita tensor we can also write this out in components, such that the Hodge star operator becomes the map \begin{equation}
    \star: \Omega^p(M)  \rightarrow \Omega^{n-p}(M): \bm{\alpha} \mapsto (\star \bm{\alpha})_{\mu_1, \ldots, \mu_{n-p}} = \frac{1}{p!}\epsilon^{\nu_1,\ldots,\nu_p}{}_{\mu_1,\ldots\mu_{n-p}} \bm{\alpha}_{\nu_1,\ldots,\nu_p}.
\end{equation}

\vspace{0.5cm}

An action $S$ can now be formulated as an integral over a Lagrangian, a function of the fields $\mathbf{L}[\phi]$ which, for every field configuration $\phi$, is an $n$-form on $M$. Using the Hodge duality we can also work with a scalar function $\star{\mathbf{L}}$ on $M$ as the Lagrangian: \begin{equation}
    S = \int_M \mathbf{L}[\phi]  = \int_M  (\star{\mathbf{L}[\phi]}) \epsilon = \int_M \mathrm{d}^4 x \sqrt{-g} (\star{\mathbf{L}[\phi]}).
\end{equation} In integrals, we will call a term a pure boundary term if it is an exact form and can thus be written as an integral over the boundary $\partial M$ using Stokes theorem, that is, if it is of the form \begin{equation}
    \int_M \mathrm{d}\mathbf{\Theta} = \int_{\partial M} \mathbf{\Theta},
\end{equation} which, using the Hodge dual $\mathbf{\Theta} = \star A$,  can be written as \begin{equation}
    \int_M \mathrm{d}\mathbf{\Theta} = \int_M \mathrm{d}^4 x \sqrt{-g} \mathcal{D}_a A^a = \int_M \mathrm{d}^4 x \partial_a (\sqrt{-g} A^a).
\end{equation} The covariant derivative for a general connection $\left(\nabla_{c} A\right)^{a_{1} \ldots a_{r}}{}_{b_{1} \ldots b_{s}}$ of a tensor density $A$ is equal to the covariant derivative of the density as if it was a tensor, plus a term $+W\Gamma_{c d }^{d} A^{a_{1} \cdots a_{r}}{}_{b_{1} \ldots b_{s}}$ with $W$ the weight of the density \cite{Heisenberg_2019}. As $\sqrt{-g}A^a = \alpha^a$ here is a vector density of weight $-1$, we have \begin{equation}
    \nabla_a \alpha^a = \partial_a \alpha^a + \Gamma^a_{ab} \alpha^b - \Gamma^b_{ab} \alpha^a = \partial_a \alpha^a + T^a_{ab} \alpha^b,
\end{equation}  which is why we can express a pure boundary term also as
\begin{equation}
     \int_M \mathrm{d}\mathbf{\Theta} = \int_M \mathrm{d}^4 x \sqrt{-g} \mathcal{D}_a A^a = \int_M \mathrm{d}^4 x \partial_a (\sqrt{-g} A^a) = \int_M \mathrm{d}^4 x \nabla_a (\sqrt{-g} A^a),
\end{equation} in terms of a torsion-free covariant derivative $\nabla$.

\vspace{0.5cm}

\end{appendices}

\newpage

\bibliography{ref}

\end{document}